\documentclass{ieeeaccess}
\usepackage{cite}
\usepackage{amsmath,amssymb,amsfonts}
\usepackage{graphicx}
\usepackage{textcomp}
\def\BibTeX{{\rm B\kern-.05em{\sc i\kern-.025em b}\kern-.08em
    T\kern-.1667em\lower.7ex\hbox{E}\kern-.125emX}}

\usepackage{dsfont}
\usepackage{amsthm}
\usepackage{color}
\usepackage{algorithm}
\usepackage{algpseudocode} 
\usepackage{comment}
\usepackage{subfig}
\usepackage{scalerel}
\usepackage{bbm}

\usepackage{siunitx}
\usepackage{textcase}

\def\NoNumber#1{{\def\alglinenumber##1{}\State #1}\addtocounter{ALG@line}{-1}}
\definecolor{mygreen}{RGB}{0,204,102}

\definecolor{gold}{RGB}{245,155,0}

\def \x {\mathbf{x}}

\def\Amc{\mathcal{A}}

\def\Emc{\mathcal{E}}

\def\Gmc{\mathcal{G}}
\def\Hmc{\mathcal{H}}

\def\Kmc{\mathcal{K}}

\def\Nmc{\mathcal{N}}

\def\Smc{\mathcal{S}}
\def\Tmc{\mathcal{T}}

\def\Vmc{\mathcal{V}}

\newtheorem{problem}{Problem}
\newtheorem{proposition}{Proposition}
\newtheorem{theorem}{Theorem}
\newtheorem{lemma}{Lemma}
\newtheorem{definition}{Definition}
\newtheorem{assumption}{Assumption}
\newtheorem{remark}{Remark}

\newcommand{\angb}[1] {\left<{#1}\right>}

\def \ones			{{\mathds{1}}} 

\begin{document}
\history{Date of publication xxxx 00, 0000, date of current version xxxx 00, 0000.}
\doi{10.1109/ACCESS.2023.0322000}

\title{Adaptive Consensus-based Reference Generation for the Regulation of Open-Channel Networks}
\author{Marco Fabris\authorrefmark{1}  \IEEEmembership{Member, IEEE}, Marco D. Bellinazzi\authorrefmark{1}, Andrea Furlanetto\authorrefmark{2}, and Angelo Cenedese\authorrefmark{1} \IEEEmembership{Senior Member, IEEE},
}

\address[1]{
Department of Information Engineering at the University of Padova, Padova, 35131. A.~Cenedese is also with the Institute of Electronics, Information and Telecommunication Engineering, National Research Council. (e-mails: marco.fabris.1@unipd.it, marcodavide.bellinazzi@studenti.unipd.it, angelo.cenedese@unipd.it)
}
\address[2]{
Consorzio Bonifica Veneto Orientale. (e-mail: andrea.furlanetto@bonificavenetorientale.it)}
\tfootnote{
}

\markboth
{M. Fabris \headeretal: Adaptive-Consensus-based Regulation of Open-Channel Networks}
{M. Fabris \headeretal: Adaptive-Consensus-based Regulation of Open-Channel Networks}

\corresp{Corresponding author: Angelo Cenedese (e-mail: angelo.cenedese@unipd.it).}

\begin{abstract}
This paper deals with water management over open-channel networks (OCNs) subject to water
height imbalance. The OCN is modeled by means of graph theoretic tools and a regulation scheme is designed basing on an outer reference generation loop for the whole OCN and a set of local controllers. 
Specifically, it is devised a fully distributed adaptive consensus-based algorithm within the discrete-time domain capable of (i) generating a suitable tracking reference that stabilizes
the water increments over the underlying network at a common level; (ii) coping with general flow
constraints related to each channel of the considered system. This iterative procedure is derived by
solving a guidance problem that guarantees to steer the regulated network - represented as a closed-loop
system - while satisfying requirements (i) and (ii), provided that a suitable design for the local
feedback law controlling each channel flow is already available. The proposed solution converges
exponentially fast towards the average consensus thanks to a Metropolis-Hastings design of the network parameters without violating the imposed constraints over
time. In addition, numerical results are reported to support the theoretical findings, and the
performance of the developed algorithm is discussed in the context of a realistic scenario.
\end{abstract}


\titlepgskip=-21pt

\maketitle

\section{Introduction}
\label{sec:introduction}
Water networks are complex large-scale systems comprising diverse components including transport structures (e.g. open channels, pipelines), flow control units, and storage apparatuses. 
In particular, open-channel networks (OCNs) serve several purposes such as draining rainwater outside urbanized areas in order to avoid flooding and ensuring an appropriate water supply for the irrigation of farmlands. 
As the frequency, intensity, and duration of storm events have increased worldwide because of climate change~\cite{XuRamanathanVictor2018,Loiy2019,LiSunHuang2020}, OCNs have proven unable to handle severe flood phenomena~\cite{PapalexiouMontanari2019} or water shortages and droughts~\cite{Vicente-SerranoQuiringPenaGallardo2020}.
Therefore, the analysis of such systems and the design of advanced control techniques to improve their management have recently become an objective of major impact and interest within the research community.

Different approaches can be developed for the water distribution problem because of the several aspects to be considered, such as cost minimization, control optimization, supply or allocation issues, leak management. 
The vast majority of solutions 
encompasses either traditional numerical modeling techniques~\cite{Szymkiewicz2010} or machine learning tools and graph-based algorithms.
As machine learning is often exploited for quality and maintenance-related issues (e.g. leak and infiltration assessment~\cite{w12041153}), graph theory is better employed for more operational planning problems.
Examples of graph theory applied to water network distribution problems are very common and used to solve different tasks. For example, in \cite{djikstra}, graph theory is used to devise an algorithm to manage the scheduling of the water network, and this 
returns an optimal minimal cost pump-scheduling pattern; more thoroughly, in \cite{w10010045} it is introduced a holistic analysis framework to support water utilities on the decision making process for efficient supply management.
Furthermore, within graph theory, it is common opinion (see \cite{10.1007/978-3-642-36071-8_1},\cite{10.1145/2678280}) that theoretical computer science lays on a vantage point for the understanding of key emergent properties in complex interconnected systems.\\ 
Overall, these trends highlight that one of the most challenging aspect within the regulation of OCNs is to find a \textit{viable}, \textit{efficient}, \textit{topology-independent} and \textit{distributed} method that can be used to solve water distribution problems pertaining to this category of networks. In this perspective, we shall model such problems and the corresponding solutions by exploiting the typical approaches employed with networked systems.

Given this premise, we build a novel solution to balance water levels in OCNs subject to floods or shortages upon the well-known consensus protocol.
In a network of agents, ``consensus'' means an agreement regarding a certain quantity of interest that usually depends on the state of all agents. 
A consensus algorithm (or protocol) is an interaction rule that specifies the information exchange between an agent and its neighbors on the network~\cite{HERLIHY2021129}.
Consensus can be applied for multiple purposes, as discussed in~\cite{4118472}, where a general theoretical framework is provided. 
Relevant to our study is~\cite{9265352}, where a consensus-based control strategy for a water distribution system is proposed. This enables a water system to continuously supply the demand while minimizing the impact of faulty equipment within the water distribution facilities.

Constraints are another aspect that consensus theory takes into account, as in~\cite{YANG2014499}, where it is considered the global consensus problem for discrete-time multi-agent systems with input saturation constraints under fixed undirected topologies. 
Due to the existence of different kinds of state constraints, most existing consensus algorithms cannot be applied directly (see e.g.~\cite{7170935}, where the state is confined in an interval around the initial conditions in order to obtain convergence on opinion dynamics and containment control). 
Hence, these distributed protocols need to be designed depending on the specific application challenge.

\textbf{Contributions}:  For the above reasons, a novel and versatile consensus algorithm is here presented to face the water level compensation subject to flow constraints in OCNs, which in fact translate into restrictions on the state variation \cite{cassan2021hydrostatic}.
Related studies concerning water distribution issues over networks can be already found in~\cite{ZamzamDallAneseZhao2019,SinghKekatos2020,SinghKekatos2021}, where power or price-based costs are minimized in order to obtain control solutions capable of ensuring optimal governance of the active elements in the underlying network. However, these strategies are not intended to operate on OCNs and deal with water leveling and flow constraints.
Differently from the methods of these works, here it is proposed an autonomous time-varying difference-equation-based model and the related regulation scheme that accounts for general flow constraints. 
In particular, we encapsulate the handling of physical restrictions affecting each waterway flow into a modified version of the classic distributed consensus protocol. Then, this approach can be carried out through the solution to a guidance problem that seeks a feasible decentralized control reference for the water exchange among the channels of the given network in order to attain a common level increment.

The main contribution of this work is thus devoted to the development of a distributed algorithm satisfying the above requirements. With the aim of optimizing water distribution in an OCN, the proposed strategy is capable of (i) allowing for the allocation of even amounts of water, in terms of height increments, over the underlying network; (ii) coping with general flow constraints associated to each one of the considered channels.
These two aspects are accommodated by resorting to the adaptation of the classic average consensus to the specific framework of interest and introducing a time-variant adjustment on the protocol, so that different water regimes occurring at each channel can be managed while reaching an agreement dictated by the mean of the initial water height increments.
Remarkably, beyond the specific application to OCNs, the proposed solution is shown to have wider application properties, namely it can be used in more general problems and  topologies, as it is designed to fit any networked systems in which the discrete-time average consensus dynamics is demanded to account for a limited capacity of information exchange.
A further contribution of this paper is indeed represented by the general exploration via Lyapunov-based convergence analysis of the devised distributed algorithm. As a result, an analytical metric for the convergence rate is also suggested.
Lastly, numerical simulations on the realistic scenario offered by the Cavallino water network (see Fig. \ref{fig:cavallinowaternetwork}) situated in the Venice metropolitan area, Italy, are reported to validate the presented theoretical findings.

\begin{figure}[t!]
	\centering
	\includegraphics[scale=0.23]{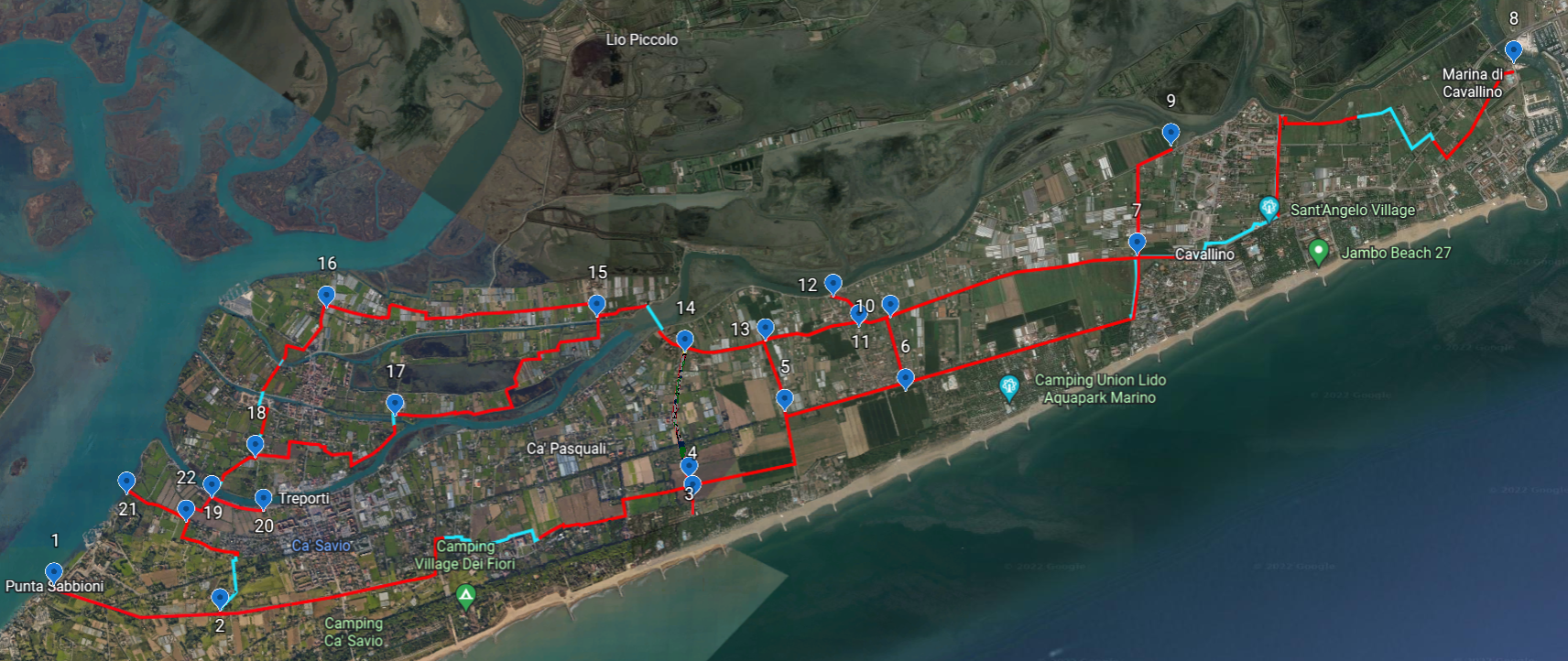}
	\caption{View of the Cavallino OCN.}
	\vspace{-6pt}
	\label{fig:cavallinowaternetwork}
\end{figure}

\textbf{Paper organization}: The remainder of this paper is organized as follows. 
In Sec.~\ref{sec:preliminaries}, we introduce mathematical preliminaries and briefly review the consensus theory. Sec.~\ref{chp:modeldesign} describes the setup for which our consensus-based reference generation protocol for even water compensation is proposed. 
To this purpose, its dynamics is presented along with a method that guarantees to avoid constraint violations corresponding to water flow limitations imposed on the network.
Then, Sec.~\ref{chp:models} yields an iterative distributed procedure for the implementation of the consensus protocol previously introduced. In relation to this, a Lyapunov-based convergence analysis is also provided in the Appendix to prove the effectiveness and the correctness of the aforementioned algorithm. 
Sec.~\ref{sec:numericalsimulations} is devoted to the results of our numerical simulations. 
Finally, Sec.~\ref{chp:conclusion} concludes our work, discussing future research directions.


\section{Preliminaries}\label{sec:preliminaries}

In this section, the preliminary notions and assumptions to model OCNs are reported. 
Also, the discrete-time consensus protocol is briefly reviewed.

\subsection{Basic notation}\label{ssec:basicnotation}
Hereafter, symbols $\mathbb{N}$, $\mathbb{R}$, $\mathbb{R}_{\geq0}$ and $\mathbb{R}_{>0}$ denote the sets of natural, real, nonnegative real, and positive real numbers.  
Both letters $k$ and $l$ indicate discrete time instants, while $t$ and $\tau$ refer to continuous time. 
With $\mathrm{Log}$ and $\mathrm{sign}$, the \textit{base-$10$ logarithm} and the \textit{sign} functions are meant.
The following notation is quite standard in linear algebra \cite{schott2016matrix}.
Given a vector $\varpi \in \mathbb{R}^{N}$ comprising of the components $\varpi_i$, with $i = 1,\ldots,N$, its \textit{infinity norm} and its \textit{span} are respectively denoted by $\left\|\varpi\right\|_{\infty}$ and $\left< \varpi \right>$. 
Given a matrix $\Omega \in \mathbb{R}^{N\times N}$, its $ij$-th entry is indicated with $[\Omega]_{ij}$ and its \textit{eigenvalues} are denoted by $\lambda^{\Omega}_{i}$, for $i = 0,\ldots,N-1$; also, by $|\Omega|$, we mean the (entry-wise) absolute value of matrix $\Omega$. 
A matrix $\Omega$ having nonnegative entries is \textit{row-stochastic} if each row sums to $1$,  $\Omega \in \mathrm{stoch}(\mathbb{R}^{n \times n})$; it is \textit{doubly-stochastic}, if it is \textit{row-stochastic} and also each column entries sum to $1$, $\Omega \in \mathrm{stoch}^{2}(\mathbb{R}^{n \times n})$.
Then, we indicate with $I_N \in \mathbb{R}^{N \times N}$ and $\ones_{N}\in \mathbb{R}^N$ the \textit{identity matrix} and the \textit{agreement vector} of dimension $N$, respectively. 
Moreover, symbols $|\Smc|$, $\top$, $\propto$ and $s$ denote respectively the \textit{cardinality} of set $\Smc$, the \textit{transpose} operation for matrices, the \textit{direct proportionality} relation and \textit{complex frequency variable} for continuous-time transfer functions.
For a continuous scalar function $f : \mathbb{R} \rightarrow \mathbb{R} : t \mapsto f(t)$ sampled with period $T_s>0$, the \textit{difference quotient} of $f$ over $T_s$ at $k$ is defined as $\delta f^{[T_s]}(k) = [f(kT_s+T_s)-f(kT_s)]/T_s$; the \textit{shift operator} is denoted with $z$, so that for all $\tau \in \mathbb{R}_{\geq 0}$ one has $f(\tau + T_s ) = z f(\tau)$. Lastly, the operator $\nabla$ is used to denote the gradient of a differentiable function.

\subsection{Graph-based OCN model}

In this work, we account for bidirectional and interconnected OCNs comprised of $n\geq 2$ channels and $m\geq 3$ junctions. The latter are of two types, and they can either link a pair of subsequent channels or simply represent an endpoint for the water system. 
A water network of this kind can be thus modeled as a \textit{graph} $\Gmc^{o} = (\Vmc^{o},\Emc^{o})$ ~\cite{GrossYellen2018}, wherein each element $v^{o}_{i}$ in the \textit{vertex set} $\Vmc^{o} = \{v^{o}_{1}, \ldots, v^{o}_{m}\}$ corresponds to a junction, and the \textit{edge set} $\Emc^{o} \subseteq \Vmc^{o} \times \Vmc^{o}$ describes each (undirected) channel\footnote{By assumption, channels are characterized by a bidirectional water flow. No restriction is imposed on the presence of hydraulic pumps capable of directing the streams along both the ways.}. Under this premise, there exists an edge $e^{o}_{\ell} := e^{o}_{ij} \in \Emc^{o}$, with $\ell = 1,\ldots,n$ and $i<j$, if and only if there exists a channel linking junctions $v^{o}_{i}$ and $v^{o}_{j}$, implying that $\Gmc^{o}$ is \textit{undirected}. In addition, it is assumed that $\Gmc^{o}$ is \textit{connected}, that is there exists a path $(e^{o}_{\ell_{i}}, \ldots, e^{o}_{\ell_{j}})$ connecting any two distinct nodes $(v_{i}^{o},v_{j}^{o}) \in \Vmc^{o} \times \Vmc^{o}$. 

Let $\mathrm{L}$ denote the \textit{line graph} operator that maps a given graph $\Hmc$ into its \textit{adjoint} $\mathrm{L}(\Hmc)$ ~\cite{Harary1960,BeinekeWilson1978}.
In order to operate with the regulation procedure proposed in this paper, the adjoint $\Gmc := \mathrm{L}(\Gmc^{o}) $ of the given topology $\Gmc^{o}$ is constructed and considered. More precisely, letting $\Gmc = (\Vmc,\Emc)$, the nodes in the (adjoint) vertex set $\Vmc = \{v_{1},\ldots,v_{n}\}$ represent each of the channels, i.e. $v_{\ell} = e^{o}_{\ell}$, for all $\ell = 1,\ldots,n$. Also, there exists an edge $e_{ij}$ in the (adjoint) edge set $\Emc \subseteq \Vmc \times \Vmc$ if and only if channels $v_{i}$ and $v_{j}$ are both incident to one of the junctions in $\Vmc^{o}$. It is well-known that if a graph $\Hmc$ is connected so is $L(\Hmc)$; hence, $\Gmc$ is connected. Furthermore, denoting with $\Nmc^{o}_{i} = \{j  ~|~ (v^{o}_i,v^{o}_j) \in \Emc^{o} \}$ the $i$-th \textit{neighborhood} of junction $i$ and with $d_{i}^{o} = |\Nmc^{o}_{i}|$ the \textit{degree} of junction $i$, the cardinality of the (adjoint) edge set is yielded by $|\Emc| = -n + \frac{1}{2} \sum_{i=1}^{m}(d^{o}_{i})^{2}$, as shown in ~\cite{Harary1969}. Similarly, we define the (adjoint) $i$-th neighborhood of channel $i$ and its corresponding (adjoint) degree as $\Nmc_{i} = \{j ~|~ (v_i,v_j) \in \Emc \}$ and $d_{i} = |\Nmc_{i}|$, respectively. Also, defining the \textit{adjacency matrix} of $\Gmc^{o}$ as $A^{o} \in \mathbb{R}^{m \times m}$ such that for each $v_{i}\in \Vmc$ it is assigned $[A^{o}]_{ij} = 1$, if $j \in \Nmc^{o}_{i}$; $[A^{o}]_{ij} = 0$, otherwise; and the \textit{incidence matrix} of $\Gmc^{o}$ as $E^{o} \in \mathbb{R}^{m \times n}$ such that for each $e_{\ell}^{o} = e_{ij}^{o} \in \Emc^{o}$ it is assigned $[E^{o}]_{\upsilon \ell} = -1$, if $\upsilon=i$; $[E^{o}]_{\upsilon \ell} = 1$, if $\upsilon=j$; $[E^{o}]_{\upsilon \ell} = 0$, otherwise. 
With these positions, it can be derived that the adjacency matrix $A \in \mathbb{R}^{n\times n}$ of $\Gmc = \mathrm{L}(\Gmc^{o})$, which is yielded by $A = |(E^{o})^{\top}E^{o} -2I_{n}|$. Moreover, we let $\overline{\Nmc}_{i} := \Nmc_{i} \cup \{i\}$, $d_{m} := \min_{i =1,\ldots,n} \{d_{i}\}$, $d_{M} := \max_{i =1,\ldots,n} \{d_{i}\}$ be respectively the extended $i$-th neighborhood, minimum degree and maximum degree in $\Gmc$. 
Lastly, the radius and diameter\footnote{The eccentricity $\varepsilon(v_{i})$ of a vertex $v_{i}$ in a connected graph $\Hmc$ is the maximum graph distance between $v_{i}$ and any other vertex $v_{j}$ of $\Hmc$.
	The radius of a graph is the minimum eccentricity of any of its vertices, while the diameter of a graph is the maximum eccentricity of any of its vertices.} of $\Gmc$ are denoted by $\rho$ and $\phi$, respectively.

\subsection{Review of the consensus protocol}
We now provide an overview of the discrete-time weighted consensus problem in the field of multi-agent systems  (see also \cite{MesbahiEgerstedt2010,Lunze2019} for more details on the topic). Let us consider a group of $n$ homogeneous agents, e.g. the $n$ pairs of actuators installed at the two endpoints of each channel in a water system modeled by an undirected and connected (adjoint) graph $\Gmc$. 
Let us also assign a discrete-time state $x_{i}(k) \in \mathbb{R}$ to the $i$-th agent, for $i = 1,\dots,N$, with $k=0,1,2,\ldots$. The full state of the whole network can be thus expressed by $x(k) = \begin{bmatrix} x_{1}(k) & \cdots & x_{n}(k) \end{bmatrix}^{\top} \in \mathbb{R}^{n}$. The discrete-time consensus within a multi-agent system can be characterized as follows.
\begin{definition}[Discrete-time consensus\textcolor{white}{,}~\cite{Bullo2019}]\label{def:consensus}
	{An $n$-agent network achieves \emph{consensus} if $\lim_{k\rightarrow+\infty} x(k) \in \Amc $, where $\Amc = \angb{\ones_{n}} $ is called the \emph{agreement set}. Moreover, if for all $i=1,\ldots,n$ it holds that $\lim_{k\rightarrow+\infty} x_{i}(k) = n^{-1} \sum_{j=1}^{n} x_{j}(0)$ then 
 \textit{average consensus} is attained.
	}
\end{definition}

Let us consider a connected graph $\Gmc=(\Vmc,\Emc)$ in which it is assigned a \textit{weight} $p_{ij} \in (0,1)$ to each edge $e_{ij} \in \Emc$ and it is assigned a \textit{self-loop} $p_{ii} \in [0,1)$ to each node $v_{i} \in \Vmc$, such that $\sum_{j\in \overline{\Nmc}_{i}} p_{ij} = 1$ for all $i=1,\ldots,n$. Let us define the \textit{update matrix} $P$ as $[P]_{ij}=p_{ij}$, if $(i,j) \in \Emc$ or $i=j$ and $[P]_{ij} = 0$, otherwise, such that $P \in \mathrm{stoch}(\mathbb{R}^{n \times n})$. 
It is well known that the \textit{linear discrete-time consensus protocol}
\begin{equation}\label{eq:LAP}
	x(k+1) = P x(k)
\end{equation}
drives the ensemble state $x(k)$ to the agreement set if at least one of the self-loops $p_{ii}$ is chosen to be strictly positive~\cite{Bullo2019}.

In many frameworks, the consensus protocol \eqref{eq:LAP} is required to perform the arithmetic mean of the initial conditions. 
For this purpose, the update matrix is usually designed to be doubly-stochastic, namely it is imposed that $P \in \mathrm{stoch}^{2}(\mathbb{R}^{n \times n})$; an example of such design can be obtained by following the Metropolis-Hastings (MH) structure in which coefficients $[P]_{ij} = p_{ij}$ are assigned as
\begin{equation}\label{metro} 
	p_{ij} := \begin{cases}
		(1+ \max(d_{i},d_{j}))^{-1}, \quad &\forall j \in \Nmc_{i}; \\
		0, \quad &\forall j \notin \overline{\Nmc}_{i} ; \\
		1- \sum\limits_{j \in \Nmc_{i}} p_{ij}, \quad &\text{otherwise}.
	\end{cases}
\end{equation}

By the Gershgorin's disk theorem~\cite{Bell1965}, $P$ has $n$ real eigenvalues $ \lambda_{n-1}^{P} \leq \cdots \leq \lambda_{1}^{P} \leq \lambda_{0}^{P} = 1$  belonging to the interval $[-1, 1]$. 
In addition, as shown in ~\cite{6854643}, it is known that this design of $P$ ensures average consensus for protocol \eqref{eq:LAP}, as it is guaranteed that $\lambda_{n-1}^{P}>-1$ and $\lambda_{1}^{P} < 1$.

\section{The role of consensus dynamics in the regulation of open-channel networks}
\label{chp:modeldesign}

The following paragraphs are devoted to the formulation of the regulation problem for an OCN. In particular, it is highlighted how the consensus dynamics can be exploited to provide an autonomous reference to balance the channels' water levels in the system.

\subsection{Setup and problem formulation}\label{ssec:setup}

\begin{figure}[b!]
	\centering
	\vspace{-6pt}
	\includegraphics[scale=0.18]{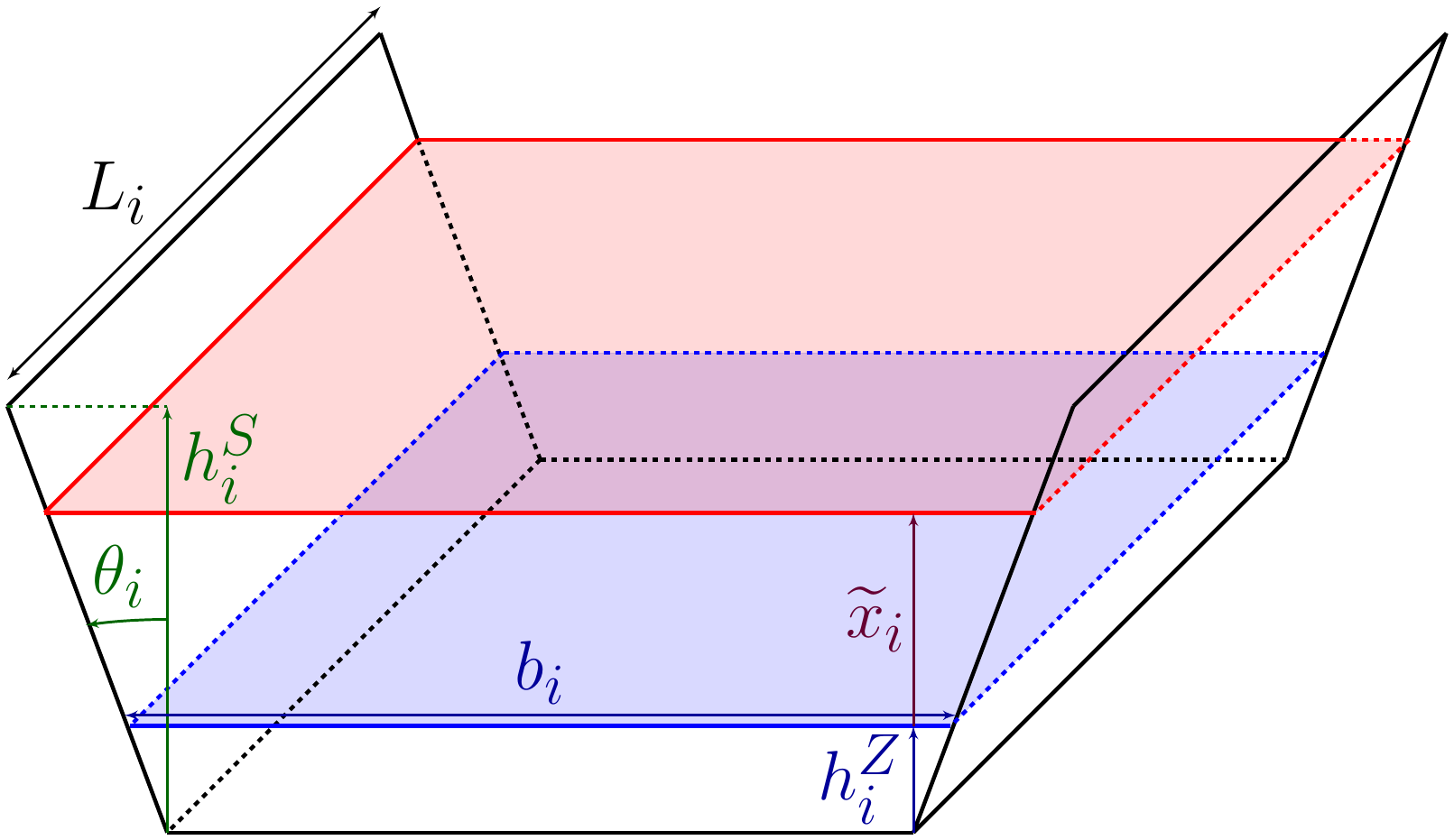}
	\caption{Geometric characterization of the $i$-th channel. The volume (increment) $V_{i}$ in \eqref{eq:volumedef} is considered as the portion of space between the blue and red surfaces. In case of $\widetilde{x}_{i}<0$ and $V_{i}<0$, the red surface lays below the blue one.}
	\label{fig:channel_dimensions}
\end{figure}

Let us consider the water channels $i=1,\ldots,n$ in $\Gmc$. As illustrated in Fig. \ref{fig:channel_dimensions}, we assume that the $i$-th channel has length $L_{i}$, trapezoidal cross section characterized by height $h_{i}^{S}>0$ and bank slope $\theta_{i} \in (0,\pi/2)$, so that $\theta_{i}$ approaching zero corresponds to having vertical banks. The $i$-th zero water level reference is chosen at height $h_{i}^{Z} \in [0,h^{S}_{i}]$, and at that level, the width of the $i$-th cross section is given by $b_{i} > 0$. Let $\widetilde{x}_{i} \in [-h_{i}^{Z},h_{i}^{S}-h_{i}^{Z}]$ be the variation of water level from reference $h^{Z}_{i}$ to $(h^{Z}_{i} +\widetilde{x}_{i}) \in [0, h_{i}^{S}]$. Then the volume variation $V_{i}(\widetilde{x}_{i})$ causing a height increment $\widetilde{x}_{i}$ is given by
\begin{equation}\label{eq:volumedef}
	V_{i} = L_{i}b_{i}\widetilde{x}_{i} + L_{i} \tan(\theta_{i}) \widetilde{x}_{i}^{2}.
\end{equation}
Because of \eqref{eq:volumedef}, since $\widetilde{x}_{i}$ is upper bounded by $\overline{x}_{i} = h^{S}_{i}-h^{Z}_{i}$, also $V_{i}$ is upper bounded by $\overline{V}_{i} =L_{i}b_{i}\overline{x}_{i} + L_{i} \tan(\theta_{i}) \overline{x}_{i}^{2} $; and since $\widetilde{x}_{i}$ is lower bounded by $\underline{x}_{i} = -h^{Z}_{i}$, also $V_{i}$ is lower bounded by $\underline{V}_{i} =L_{i}b_{i}\underline{x}_{i} + L_{i} \tan(\theta_{i}) \underline{x}_{i}^{2} $.
Consequently, by inverting \eqref{eq:volumedef}, we find the dependence of increment $\widetilde{x}_{i}$ w.r.t. the corresponding volume variation $V_{i}$, that is,
\begin{equation}\label{eq:heightdef}
	\widetilde{x}_{i} = (a_{i}^{\prime}+b_{i}^{\prime}V_{i})^{1/2}-c_{i}^{\prime},
\end{equation}
where $c_{i}^{\prime} = b_{i}/(2\tan (\theta_{i}))$, $b_{i}^{\prime} = (L_{i}\tan (\theta_{i}))^{-1}$ and $a_{i}^{\prime} = (c_{i}^{\prime})^{2}$ are positive constants depending on the geometry of the $i$-th channel. 

Now, let us consider the difference quotients $\delta \widetilde{x}_{i}^{[T_s]}(k)$, $\delta V_{i}^{[T_s]}(k)$ at time $k$ and define the \textit{download and upload limitations} for the $i$-th water flow rate as the functions 
\begin{equation}\label{eq:Cfunctions}
	C^{J}_{i} : \mathbb{N} \rightarrow \mathbb{R}_{>0} : k \mapsto C^{J}_{i}(k), ~ J\in \{D,U\},
\end{equation}
with each $C_{i}^{J}(k)$ bounded from above. By means of \eqref{eq:Cfunctions}, the following flow rate constraint can be taken into account:
\begin{equation}\label{eq:volumeconstraints}
	-C_{i}^{D}(k) \leq \delta V_{i}^{[T_s]}(k) \leq C_{i}^{U}(k) , \quad \forall k = 0,1,2,\ldots.
\end{equation}
In relation to the $i$-th channel, it is easy to show that if there exist functions
\begin{equation}\label{eq:cfunctions}
	c^{J}_{i} : \mathbb{N} \rightarrow \mathbb{R}_{>0} : k \mapsto c^{J}_{i}(k), ~ J\in \{D,U\},
\end{equation}
bounded from above, and the water height constraint
\begin{equation}\label{eq:heightconstraints}
	c_{i}^{D}(k) \leq \delta \widetilde{x}_{i}^{[T_s]}(k) \leq c_{i}^{U}(k) , \quad \forall k = 0,1,2,\ldots
\end{equation}
is enforced, then flow rate constraint in \eqref{eq:volumeconstraints} is guaranteed, provided that $C_{i}^{J}(k) $ and $c_{i}^{J}(k) $ are suitably related to each other. 
In particular, by \eqref{eq:heightdef}, constraint \eqref{eq:heightconstraints} can be rewritten as
\begin{align}\label{eq:heightconstraints2} 
	-c_{i}^{D}(k) \leq  &T_s^{-1} [(a_{i}^{\prime}+b_{i}^{\prime}V(kT_s+T_s))^{1/2}  \nonumber \\
	&- (a_{i}^{\prime}+b_{i}^{\prime}V(kT_s))^{1/2}] \leq c_{i}^{U}(k) .
\end{align}
Multiplying each term in \eqref{eq:heightconstraints2} by $\underline{w}_{i}(k) := (b_{i}^{\prime})^{-1}  [(a_{i}^{\prime}+b_{i}^{\prime}V_{i}(kT_s+T_s))^{1/2} + (a_{i}^{\prime}+b_{i}^{\prime}V_{i}(kT_s))^{1/2}] >0$ one obtains
\begin{align}\label{eq:volumecontraints2}
	-w_{i} c_{i}^{D}(k) &\leq -\underline{w}_{i}(k) c_{i}^{D}(k)  \nonumber \\
	&\leq \delta V_{i}^{[T_s]}(k) \leq \underline{w}_{i}(k) c_{i}^{U}(k) \leq w_{i} c_{i}^{U}(k) ,
\end{align}
and, clearly, if condition $C_{i}^{J}(k) \geq w_{i} c_{i}^{J}(k)$, $J \in \{U,D\}$, is verified for all $k=0,1,2,\ldots$ with $w_{i} = 2(b_{i}^{\prime})^{-1}  (a_{i}^{\prime}+b_{i}^{\prime}\overline{V}_{i})^{1/2} \geq \underline{w}_{i}(k)$, then \eqref{eq:volumecontraints2} becomes a stricter or equivalent version of \eqref{eq:volumeconstraints}. Inequality \eqref{eq:volumecontraints2} can be trivially extended to the case in which $\theta_{i}=0$ by imposing $w_{i}=\underline{w}_{i}=L_{i}b_{i}$.

\begin{figure*}[t!]
	\centering
	\includegraphics[scale=0.26]{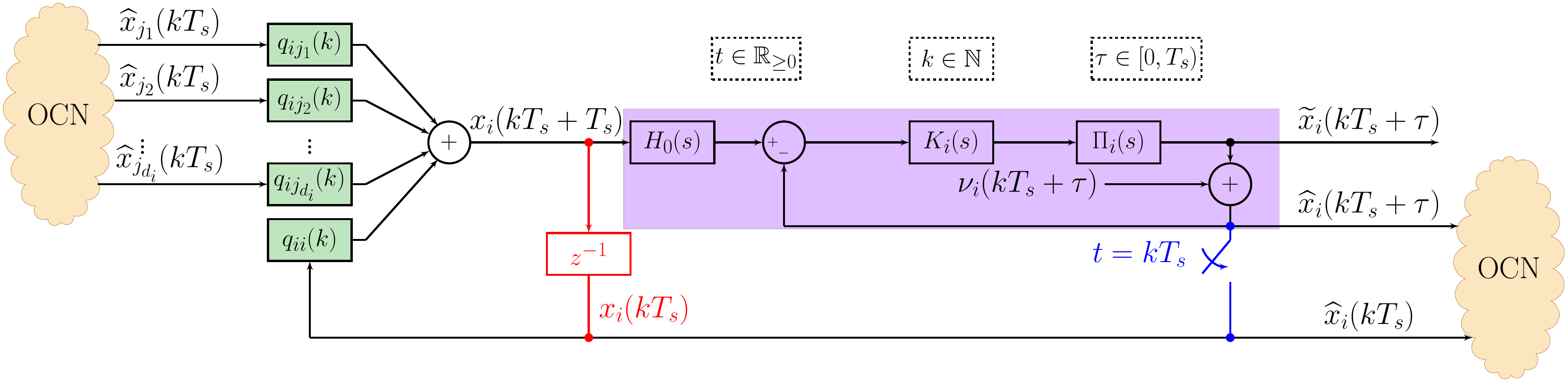}
	\caption{Local control scheme for the $i$-th channel. 
 The inner control loop is shaded in violet and the time-varying consensus coefficients $q_{ij}(k)$ though which the previous value of the generated reference is combined among $i$'s neighbors are highlighted in green (see Subsec. \ref{subsec:proposedRDP} for a detailed discussion). The rest of the OCN is schematically represented as the orange cloud. 
With the main aim of regulating the OCN, the information flowing in the red element is considered, replacing that coming from the blue element, which bridges the inner and outer feedback loops and is used in the practical implementation of the scheme.
	}
	\label{fig:ControlSchemeCavallinoDraft}
\end{figure*}

Under closed-loop control, local measurements of the system physical quantities play a fundamental role in the design of a regulator, i.e. a scheme that estimates the current state and governs it. In the considered framework, height measurements are likely more reliable and accessible than volume ones because either the latter are derived from the former and thus subject to larger errors, or more complex measuring techniques are needed to retrieve volume information.
Thus, motivated by the fact that water flow constraint in \eqref{eq:volumeconstraints} can be ensured by inequality \eqref{eq:heightconstraints} on the water height change rate through the selection of each $c_{i}^{J}(k)$, such that $c_{i}^{J}(k) \leq C_{i}^{J}(k)/w_{i}$, we propose a general regulation scheme that evenly balance the height references across the given network by acting on the local water levels of the channels. 
To this aim, we suppose to deploy $n$ agents on the water system, namely $n$ pairs of communicating actuators (e.g. weir gates or pumps)  installed at the endpoints of each channel. 
More formally, the following assumption is made.

\begin{assumption}\label{asm:network}
	The flow of the given water network $\Gmc^{o}$ is controlled by $n$ agents (actuator pairs), each one of them associated to the corresponding $i$-th channel. 
	The communication established through the network of agents is then captured by topology $\Gmc = \mathrm{L}(\Gmc^{o})$. 
\end{assumption}

In order to focus on the design of a water distribution algorithm, we further assume that the $i$-th agent is endowed with a local closed-loop control scheme capable of regulating the water increment $\widetilde{x}_{i}$ exactly. 
As schematically shown in Fig.~\ref{fig:ControlSchemeCavallinoDraft}, such a scheme is formed by each inner feedback loop (shaded in violet) determined by the $i$-th inner controller $K_{i}(s)$ and the $i$-th plant $\Pi_{i}(s)$ and an outer feedback loop (in black). 

More precisely, letting $\nu_i(t)$ be a zero mean Gaussian white noise and given the $i$-th measured water level increment $\widehat{x}_{i}(kT_s) = \widetilde{x}_i(kT_s) + \nu_i(kT_s)$ at time $t=kT_s$, the reference $x_{i}(kT_s + T_s)\in \mathbb{R}$ is imposed and is supposed to be ideally tracked by the true unknown state\footnote{Such an ideal estimation is considered just to simplify the discussion in the sequel. Nonetheless, filtering techniques should be exploited in practice to compensate for measurement noise, while properly designed local controllers allow reaching the reference value.}, i.e. it holds that $\widetilde{x}_i(k T_s + \tau) \rightarrow x_i(kT_s + T_s)$ as $\tau \rightarrow T_s$.
Considering the channel $i$ and the corresponding control depicted in Fig. \ref{fig:ControlSchemeCavallinoDraft}, the current reference $x_{i}(k T_s+T_s)$ is consequently computed as a linear combination of the past input references $x_{j}(k T_s)$, with $j \in \overline{\Nmc}_{i} = \{i,j_{1},\ldots,j_{d_{i}}\}$, via time-varying coefficients $q_{ij}(k)$.  
In addition, it is worth to mention that, for a real setting, the current reference $x_{i}(k T_s+T_s)$ has to be maintained over time for all $t \in [kT_s,kT_s+T_s)$, e.g. by means of the classic zero-order hold $H_{0}(s)$. 

Also, a final setup assumption regards the initial conditions.

\begin{assumption}\label{ass}
	The mean of the initial conditions $\alpha := n^{-1}\sum_{i=1}^{n} \widehat{x}_{i}(0) $ is equal to zero.
\end{assumption}
Indeed, whenever $\alpha \neq 0$ occurs then the ensemble state $\widehat{x}(0)$ can be detrended and reassigned so that $\widehat{x}_{i}(0) \gets (\widehat{x}_{i}(0) - \alpha)$ for all $i= 1,\ldots,n$. It is worth to note that this operation is just equivalent to recalculate each reference $h_{i}^{Z}$ as $h_{i}^{Z} \gets h_{i}^{\star} : = h_{i}^{Z}+\alpha$, where $\alpha$ can be computed by running a preliminary average consensus protocol, e.g. through \eqref{eq:LAP}-\eqref{metro}. Preprocessing the initial data via detrending as said is advantageous because each quantity $h_{i}^{\star}$ provides the desired final water level.

Denoting with $x \in \mathbb{R}^n$ the ensemble reference vector, the even compensation of the water levels in the underlying OCN can be thus formulated as the (decentralized) minimization of the following objective function:
\begin{equation}\label{eq:costfunction}
	J(x) :=  \frac{1}{2} \sum\limits_{i=1}^{n} \sum\limits_{j\in \Nmc_{i}} p_{ij}(x_{i}-x_{j})^{2} = \frac{1}{2} x^{\top} (I_{n}-P) x,
\end{equation}
where it is imposed for the coefficients $p_{ij}$ of the matrix $P$ to have MH characterization, as in \eqref{metro}.
We thus finally formalize the following \textit{guidance} problem.
\begin{problem}\label{problem}
	Suppose to govern the flow dynamics of the underlying OCN via the control scheme depicted in Fig. \ref{fig:ControlSchemeCavallinoDraft} and let assumptions Asm. \ref{asm:network} - Asm. \ref{ass} be satisfied. 
	Design an iterative and fully distributed discrete-time procedure that determines an update rule for the $i$-th increment reference $x_{i}(k T_s + T_s) $, which is tracked by (an estimate of) the $i$-th state $\widetilde{x}_{i}(kT_s)$ over $k \in \mathbb{N}$, with $i = 1, \ldots , n$. In particular, ensure that $x(kTs)$  minimizes \eqref{eq:costfunction} while constraint in \eqref{eq:heightconstraints} is guaranteed to hold for all $k \in \mathbb{N}$.
\end{problem}
Practically speaking, solving Problem \ref{problem} corresponds to the design of a dynamic reference for the local controllers, which guarantees to minimize the imbalance of water levels throughout the OCN - possibly leading to their equalization - starting from an arbitrarily uneven initial level distribution (due, e.g., to localized natural events or human actions or system failure).
Note that Problem \ref{problem} does not deal with the design of a specific control strategy, since the goal of this paper is to construct a reference signal for the regulation of an OCN. In fact, \textit{any} valid control law allowing for Asm. \ref{asm:network} - Asm. \ref{ass} that seeks reference tracking can be employed.  Rather, Problem \ref{problem} focuses on yielding an operative sequence that (i) serves as a water increment reference to be tracked; (ii) optimally compensates the water levels in the underlying OCN, namely it directly minimizes objective \eqref{eq:costfunction}; (iii) does not violate the water height constraint \eqref{eq:heightconstraints}.

\subsection{Proposed consensus-based reference generation protocol for OCN regulation} \label{subsec:proposedRDP}
In the sequel, we present the dynamics of the considered iterative water level regulation scheme introduced in Sec.~\ref{ssec:setup} in order to find a solution to Problem \ref{problem}. 
Drawing inspiration from~\cite{fabris2019distributed}, we include memory in the classic consensus dynamics~\eqref{eq:LAP} to provide a weighted correction to current water level, thus leading to the so-called \textit{reference generation protocol}~(RGP)
\begin{align} 
	x(k+1)&=\eta(k) x(k) + (1-\eta(k))Px(k)  \nonumber\\ 
	&=(\eta(k) I_{n} + (1-\eta(k))P)x(k) :=Q_{\eta}(k)x(k),
	\label{eq:3}
\end{align}
where $x \in \mathbb{R}^n$ represents the ensemble water increment reference, $T_s = 1$ is set w.l.o.g. and $\eta(k) \in (0, 1)$ can be considered a parameter trading-off self-measurements against neighbors' measurements. Imposing the MH structure on matrix $P$ in \eqref{eq:3}, as specified in \eqref{metro}, a crucial observation immediately follows.
\begin{remark}
	Leveraging the structure of objective $J(x)$ in \eqref{eq:costfunction}, dynamics in \eqref{eq:3} can be easily rewritten as 
	\begin{equation}\label{eq:updaterulesd}
		x(k+1) = x(k) - (1-\eta(k)) ~ \nabla J(x(k)).
	\end{equation}
	The expression \eqref{eq:updaterulesd} can be considered a steepest descent update rule \cite{Argyros2022} with adaptive step-size $(1-\eta(k))$ and descent direction $-\nabla J(x(k)) = (P-I_{n})x(k)$. Therefore, the proposed RGP \eqref{eq:3} can be seen as a distributed direct method to minimize the objective \eqref{eq:costfunction}. Notice that, in principle, several approaches can be adopted as an alternative of \eqref{eq:updaterulesd} to render the objective $J(x)$ minimum \cite{NedicLiuDistrOpt4Control2018}, as far as the local controller employed is able to track the generated reference.
\end{remark}

It is worth to note that matrix $Q_{\eta}(k)  = (\eta(k) I_{n} + (1-\eta(k))P) \in \mathrm{stoch}^{2}(\mathbb{R}^{n \times n})$ is still doubly-stochastic, as its entry is given by $[Q_{\eta}(k)]_{ij} := q_{ij}(k) = (1-\eta(k))p_{ij}$, if $i\neq j$; $[Q_{\eta}(k)]_{ii} := q_{ii}(k) = \eta(k)+(1-\eta(k))p_{ii}$, otherwise. 
Its eigenvalues at time $k$ belong to the interval $(-1 + 2\eta(k), 1]$; indeed, exploiting the linearity of the spectrum, it holds that $\lambda_{i}^{Q_{\eta}}(k) = \eta(k) + (1 - \eta(k))\lambda_{i}^P $, for all $ i = 0, \ldots, n - 1$. Also, parameter $\eta(k)$ allows the presence of positive self-loops.

Intuitively, the $\eta(k)$ parameter can be tuned to control the dynamics in~\eqref{eq:3}, since its convergence is a function of the spectrum of $Q_{\eta}(k)$, which is strongly dependent on $\eta(k)$.
A good and viable strategy of selecting parameter $\eta(k)$ when it is constant, namely if $\eta(k)=\eta_{0}$ and $Q_{\eta}(k) = Q_{\eta_{0}}$ for all $k\in \mathbb{N}$, is given by the minimization of the second largest (in modulus) eigenvalue of $Q_{\eta_{0}}$:
\begin{equation}\label{eq:optimaletamethod}
	\eta^{\star}=\underset{\eta_{0} }{\arg \min} \left\lbrace  \underset{i = 1, \ldots,n-1}{\max} \{| \lambda_{i}^{Q_{\eta_{0}}}| \}\right\rbrace.
\end{equation}
As shown in~\cite{FabrisMichielettoCenedese2022}, after assigning
\begin{equation}\label{eq:varsigmaPdef}
	\varsigma_{P} := (\lambda_{1}^{P}+\lambda_{n-1}^{P})/2,
\end{equation}
the optimal value for $\eta_{0}$ is indeed yielded by
\begin{equation}\label{eq:3.3}
	\eta^\star= \varsigma_{P}/ (\varsigma_{P}-1).  
\end{equation}
However, it is well-known that spectral analysis applied to time-varying dynamical systems cannot be exploited to study their convergence. As method in \eqref{eq:optimaletamethod} cannot be used in this setting, in the subsequent discussions we provide a suitable approach to design the value of $\eta(k)$ at each time $k$ -- also accounting for water flow constraints -- and an appropriate convergence analysis concerning \eqref{eq:3} is treated in App.~\ref{app:convergenceanalysis}.

\begin{remark}
		It is well-known that undesired perturbations affecting the coupling consensus weights w.r.t. a nominal value may lead to instability for the whole interconnected system \cite{FabrisZelazo2022,FabrisZelazo2023}. Consequently,
		besides the need to cope with constraints depending on the network state, an effective design of $\eta(k)$ at each instant $k$ is crucial to guarantee fast and robust convergence properties for the protocol in \eqref{eq:3}. 
\end{remark}

\subsection{Handling of the network constraints}

The formulation of system \eqref{eq:3} defined through the update matrix $Q_{\eta}(k)$ 
does not take into account limitations to the information exchange between two connected nodes (i.e., channels).
To this purpose, a proper tuning for parameter $\eta(k)$ is proven to ensure that constraint~\eqref{eq:heightconstraints} is satisfied for all $k=0,1,2,\ldots$, so that water flows can be desirably handled.
Again, w.l.o.g. we let $T_s=1$, requiring the following local constraints to hold at each iteration $k$ in relation to the water level variation $\delta x_{i}(k) := \delta x_{i}^{[1]}(k)$:
\begin{enumerate}
	\item[(i)] if $ \delta x_{i}(k) < 0$, we say that node $i$ is in \textit{download regime}, so the $i$-th \textit{download constraint} holds
	\begin{equation}
		\delta x_{i}(k)\geq -c_{i}^{D}(k) ;
		\label{eq:7}
	\end{equation}
	\item[(ii)] if $ \delta x_{i}(k)>0$, we say that node $i$ is in \textit{upload regime}, so the $i$-th \textit{upload constraint} holds
	\begin{equation}
		\delta x_{i}(k)\leq c_{i}^{U}(k);
		\label{eq:8}
	\end{equation}
	\item[(iii)] otherwise, node $i$ is considered to be \textit{at the equilibrium}, and we simply allow
	\begin{equation}
		\delta x_{i}(k)=0.
		\label{eq:equi}
	\end{equation}
\end{enumerate}
Clearly, the download constraint in~\eqref{eq:7} regulates the outgoing flow of a node towards its neighbors. On the other hand, the upload constraint in~\eqref{eq:8} accounts for the opposite effect, namely the capacity for a node to receive an incoming water flow from its neighbors. Also, for the sake of completeness, \eqref{eq:equi} specifies the case relative to the equilibrium regime\footnote{Notice that if \eqref{eq:equi} holds for all $i=1,\ldots,n$ the network consensus is achieved as desired, since the agreement vector $\alpha^{\prime} \ones_{n}$, $\alpha^{\prime} \in \mathbb{R}$, represents the sole equilibrium for system \eqref{eq:3}.} for node $i$, that is $x_{i}(k+1)=x_{i}(k)$.

Now, a proper tuning for the parameter $\eta(k)$ is derived to ensure \eqref{eq:7}-\eqref{eq:equi} during the execution of protocol \eqref{eq:3}.
To begin, we examine the download regime: under this condition, we provide a value for $\eta(k)$ that is denoted with $\eta_{D}(k)$.\\
By considering the $i$-th equation of RGP~\eqref{eq:3} and substituting $x_{i}(k+1)$ into the download constraint \eqref{eq:7} one has:
\begin{equation}
	\eta_{D}(k)\left(\!\!x_{i}(k)-\sum_{j=1}^n p_{ij}x_{j}(k)\!\!\right) \geq x_{i}(k)-\sum_{j=1}^n p_{ij}x_{j}(k)-c_{i}^{D}(k).
	\label{eq:11} 
\end{equation}
In order to state an explicit relation for $\eta_{D}(k)$, we first observe that $(x_{i}(k)-\sum_{j=1}^n p_{ij}x_{j}(k))$ is strictly positive $\forall i = 1,\ldots,n$; this is proven by substituting $x_{i}(k+1)$, as expressed in RGP~\eqref{eq:3}, into the download regime characterization~$\delta x_{i}(k)<0$, and by exploiting the fact that $\eta_{D}(k) \in (0,1)$ will be guaranteed.
Hence, one obtains
\begin{equation}
	\eta_{D}(k) \geq 1-\frac{c_{i}^{D}(k)}{x_{i}(k)-\sum_{j=1}^n p_{ij}x_{j}(k)}.
	\label{eq:12}
\end{equation}
Clearly, inequality \eqref{eq:12} exhibits a direct dependence on the $i$-th state; thus, it cannot be properly used to derive parameter $\eta_{D}(k)$, as this is a global quantity. Nonetheless, enforcing
\begin{equation}
	\eta_{D}(k) \geq 1- \frac{c_{D}(k)}{\left\|x(k)-Px(k)\right\|_{\infty}},
	\label{eq:3.16}
\end{equation}
where 
\begin{equation}\label{eq:cD}
	c_{D}(k) := \min_{i=1,\ldots,n} \left\lbrace c_{i}^{D}(k) \right\rbrace,
\end{equation}
inequality \eqref{eq:12} is satisfied for all $i=1,\ldots,n$. Notice however that with these positions,
\eqref{eq:3.16} expresses a tight (centralized) upper bound for \eqref{eq:12}.\\
As a matter of fact, aiming at distributing the computation of $\eta_{D}(k)$, we exploit the fact that
\begin{equation}\label{eq:omegainequalitydef}
	\left\|x(k)-Px(k)\right\|_{\infty} \leq \left\|I_{n}-P\right\|_{\infty}\left\|x(k)\right\|_{\infty}  \leq \omega  \left\|x(k)\right\|_{\infty}
\end{equation}
holds by the submultiplicative property of the infinite norm and Lem.~\ref{omeg} in App. \ref{app:prelimlemmas} (with $\omega$ defined as in \eqref{eq:omegadef}). It is thus possible to impose
\begin{equation}\label{eq:3.27}
	\eta_{D}(k) \geq 1- \frac{c_{D}(k)}{\omega\left\|x(k)\right\|_{\infty}}, \quad \text{if } \left\|x(k)\right\|_{\infty}>0. 
\end{equation}
Differently from \eqref{eq:3.16}, inequality \eqref{eq:3.27} is more conservative but it allows to provide a fully distributed design method for $\eta(k)$, e.g. by relying only on the so-called \textit{max-consensus} protocol to retrieve the value of $\left\|x(k)\right\|_{\infty}$.\\
Remarkably, the r.h.s. of \eqref{eq:3.27} can be used to set the value of $\eta_{D}(k)$.
However, in general, it is not guaranteed for quantity $(1-c_{D}(k)/(\omega\left\|x(k)\right\|_{\infty}))$ to be strictly positive. For this reason, we introduce a parameter
\begin{equation}\label{etal}
	\eta_L=\begin{cases} \eta^{\star} \quad & \text{if } \eta^{\star}>0; \\
		\zeta \quad &\text{otherwise};
	\end{cases}
\end{equation}
where $\eta^{\star}$ is chosen as in \eqref{eq:3.3} and $\zeta \in (0,1)$ is an arbitrarily small given constant\footnote{To the authors' experience, it seems that, actually, all $P \in \mathrm{stoch}^{2}(\mathbb{R}^{n \times n})$ defined via the MH in method \eqref{metro} yield $\varsigma_{P} \geq 0$ (see \eqref{eq:varsigmaPdef}). This leads to the necessity of imposing $\eta_{L} = \zeta$, within this framework, discarding systematically the more desirable choice $\eta_{L} = \eta^{\star}$.}. 
Quantity $\eta_{L}$, defined as in \eqref{etal}, indeed prevents to obtain $\eta_{D}(k) \le 0$ by setting $\eta_{D}(k) = \eta_{L}$ if $1- c_{D}(k) / (\omega\left\|x(k)\right\|_{\infty}) \leq 0$.
On the one hand, $\eta_{L}$ represents a suboptimal choice for the value of $\eta_{D}(k)$ whenever $\eta_{L} \geq 1- c_{D}(k) / (\omega\left\|x(k)\right\|_{\infty})$, since it ensures fast convergence for the corresponding static\footnote{Namely adopting a constant $\eta(k)=\eta_{L}$ in RGP \eqref{eq:3}.} consensus protocol. 
On the other hand, selecting $\eta_{D}(k)=\eta_{L}$ whenever $\eta_{L} < 1- c_{D}(k) / (\omega\left\|x(k)\right\|_{\infty})$ may result in the violation of download constraint \eqref{eq:7}. Therefore, we finally set
\begin{equation}\label{eq:lastineqetaD}
	\eta_{D}(k) =
	\begin{cases}
		\max \left(\eta_L, 1- \frac{c_{D}(k)}{\omega\left\|x(k)\right\|_{\infty}}\right), \quad & \text{if } \left\|x(k)\right\|_{\infty}>0; \\
		\eta_{L}, \quad &\text{otherwise}.
	\end{cases}
\end{equation}

Whereas, for the upload regime, we provide a value for $\eta(k)$ that is denoted with $\eta_{U}(k)$ through a similar reasoning. 
Setting

\begin{equation}\label{eq:cU}
	c_{U}(k) := \min_{i=1,\ldots,n} \left\lbrace c_{i}^{U}(k) \right\rbrace,
\end{equation}
we retrieve
\begin{equation}\label{eq:lastineqetaU}
	\eta_{U}(k) =
	\begin{cases}
		\max \left(\eta_L, 1- \frac{c_{U}(k)}{\omega\left\|x(k)\right\|_{\infty}}\right), \quad &\text{if }\left\|x(k)\right\|_{\infty}>0;\\
		\eta_{L}, \quad &\text{otherwise};
	\end{cases}
\end{equation}
where $\omega$ and $\eta_{L}$ respectively defined as in \eqref{eq:omegadef} and \eqref{etal}.

Combining the results obtained for the download and upload regimes in \eqref{eq:lastineqetaD} and \eqref{eq:lastineqetaU}, the choice yielded by $\eta(k) := \max(\eta_{D}(k),\eta_{U}(k))$ now appears straightforward in order to guarantee constraints \eqref{eq:7}-\eqref{eq:8} to hold for all $k=0,1,2,\ldots$. 
As a consequence, assigning
\begin{equation}\label{eq:c_min}
	c(k) := \min(c_{D}(k),c_{U}(k)),
\end{equation}
with $c_{D}$ and $c_{U}$ respectively defined as in \eqref{eq:cD} and \eqref{eq:cU}, we finally set 
\begin{equation}\label{eq:3.33}
	\eta(k) = 
	\begin{cases}
		\max \left(\eta_L, 1- \frac{c(k)}{\omega \left\|x(k)\right\|_{\infty}}\right), \quad &\text{if } \left\|x(k)\right\|_{\infty}>0 \\
		\eta_{L}, \quad &\text{otherwise}.
	\end{cases}
\end{equation} 
The next proposition demonstrates that the expression of $\eta(k)$ just provided is admissible, namely it is ensured that $\eta(k) \in (0,1)$ and $x_{i}(k) \in [-h_{i}^{Z},h_{i}^{S}-h_{i}^{Z}]$, $\forall k = 0,1,2,\ldots$, as required.

\begin{proposition}\label{prop:etabounds}
	Let us assign 
	\begin{equation}\label{eq:etaHdef}
		\eta_H:=
		\begin{cases}
			\max \left( \eta_L, 1-\frac{\underset{k=0,1,2,\ldots}{\min}\{c(k)\}}{\omega\left\|x(0)\right\|_{\infty}}\right) , \quad &\text{if } \left\|x(0)\right\|_{\infty}>0 \\
			\eta_{L}, \quad &\text{otherwise}.
		\end{cases}
	\end{equation}
	with $\omega$, $\eta_{L}$ and $c(k)$ respectively defined as in \eqref{eq:omegadef}, \eqref{etal} and \eqref{eq:c_min}. Then, for~$\eta(k)$ defined as in~\eqref{eq:3.33}, one has
	\begin{equation}
		\label{etadom}
		\eta(k) \in [\eta_L,\eta_H]\subseteq (0,1), \quad \forall{k}=1,2,3 \ldots.
	\end{equation} 
	Moreover, if $x_{i}(0) \in [-h_{i}^{Z},h_{i}^{S}-h_{i}^{Z}]$, then for all $k = 1,2,\ldots$ it holds that $x_{i}(k) \in [-h_{i}^{Z},h_{i}^{S}-h_{i}^{Z}]$, with $ i = 1,\ldots,n$.
\end{proposition}
\begin{proof}
	From the definition of $\eta(k)$ given in \eqref{eq:3.33} we can trivially deduce that $\eta(k) \geq \eta_L>0$.
	Regarding the upper limit of \eqref{etadom}, we first observe that
	\begin{equation}
		\max_{k=0,1,2,\ldots} \{\left\|x(k)\right\|_{\infty}\}=\left\|x(0)\right\|_{\infty},
		\label{eq:upper}
	\end{equation}
	since the update rule in \eqref{eq:3} is structured as a convex combination of the previous state entries.
	Then, expression \eqref{eq:etaHdef}, the positivity of functions $c_{i}^{J}$, $J \in \{D,U\}$, and equation~\eqref{eq:upper} imply that
	$\eta(k) \leq \eta_{H} <1 $, $\forall{k}=1,2,3 \ldots$.
	Lastly, thanks to \eqref{eq:upper}, $x_{i}(k) \in [-h_{i}^{Z},h_{i}^{S}-h_{i}^{Z}]$ is satisfied $\forall k = 1,2,\ldots$ whenever $x_{i}(0) \in [-h_{i}^{Z},h_{i}^{S}-h_{i}^{Z}]$.
\end{proof}

Notably, provided that each node in the network is aware of constants $\eta_{L}$, $\omega$ and the prescribed constraint $c(k)$, for all $k=0,1,2,\ldots$, expression \eqref{eq:3.33} can be computed in a fully distributed fashion. In particular, the value of $\left\|x(k)\right\|_{\infty}$ can be determined through the max-consensus protocol (MCP) ~\cite{Cortes2008,GianniniDiPaolaPetittiRizzo2013,AbdelrahimHendrickxHeemels2017}. Moreover, expression \eqref{eq:3.33} guarantees by design that constraints \eqref{eq:7}-\eqref{eq:8} are not violated at any instant $k\geq 0$.

\section{Distributed implementation of the proposed reference generation protocol}
\label{chp:models}
In this section, we present the main contribution of this work, namely a new distributed strategy that leverages time-varying consensus to evenly compensate water levels, along with its convergence analysis.

Generally speaking, the distribution imbalance of a quantity among states can be measured and described in many ways, such as $J(x)$ in \eqref{eq:costfunction}. However, to discuss convergence properties, here we define the distribution imbalance through the \textit{max-min disagreement function} $W: \mathbb{R}^n \rightarrow \mathbb{R}_{+} : x(k) \mapsto W(x(k))$ as
\begin{equation}\label{eq:lyap}
	W(x(k))= \underset{i=1,\ldots,n}{\max} \{x_i(k)\}-\underset{i=1,\ldots,n}{\min} \{x_i(k)\},
\end{equation}
such that $W(x(k))=0$ if and only if $x(k) = \alpha^{\prime} \ones_{n}$, with $\alpha^{\prime} \in \mathbb{R}$. So, to establish whether the consensus is reached, we select an arbitrarily small threshold $\gamma > 0$ for which condition $W(x(k)) \leq \gamma$ allows us to state that $x(k) \simeq \alpha^{\prime} \ones_{n}$. Also, observe that $\underline{\xi} W(x)^{2} \leq J(x) \leq (n/2) W(x)^{2}$ holds for all $x\in \mathbb{R}^{n}$, with $\underline{\xi}$ defined as in \eqref{eq:minentryP}. Hence, any convergence performance expressed in function of the objective $J(x)$ in \eqref{eq:costfunction} can be bounded through the properties of $W(x)$ in \eqref{eq:lyap}. 

Under this premise, our aim is to reach an agreement relatively to the network states in order to ensure even water distribution by minimizing $W(x(k))$ in \eqref{eq:lyap}. 
This minimization is attained through Alg.~\ref{alg:1}, hereafter discussed, which indeed terminates at $\bar{k}$ as soon as
\begin{equation}\label{eq:terminationcond1}
	W(x(\bar{k})) \leq \gamma, \quad  \text{ for the smallest } \bar{k} \geq 0, 
\end{equation}
assuring, as formally demonstrated in App. \ref{app:convergenceanalysis},
\begin{equation}\label{eq:terminationcond2}
	x_{i}(\bar{k}) \simeq \alpha, \quad \forall i=1,\ldots,n;
\end{equation}
with $\alpha$ defined in Asm. \ref{ass}. 

\begin{figure*}[t!]
	\centering
	\includegraphics[scale=0.26]{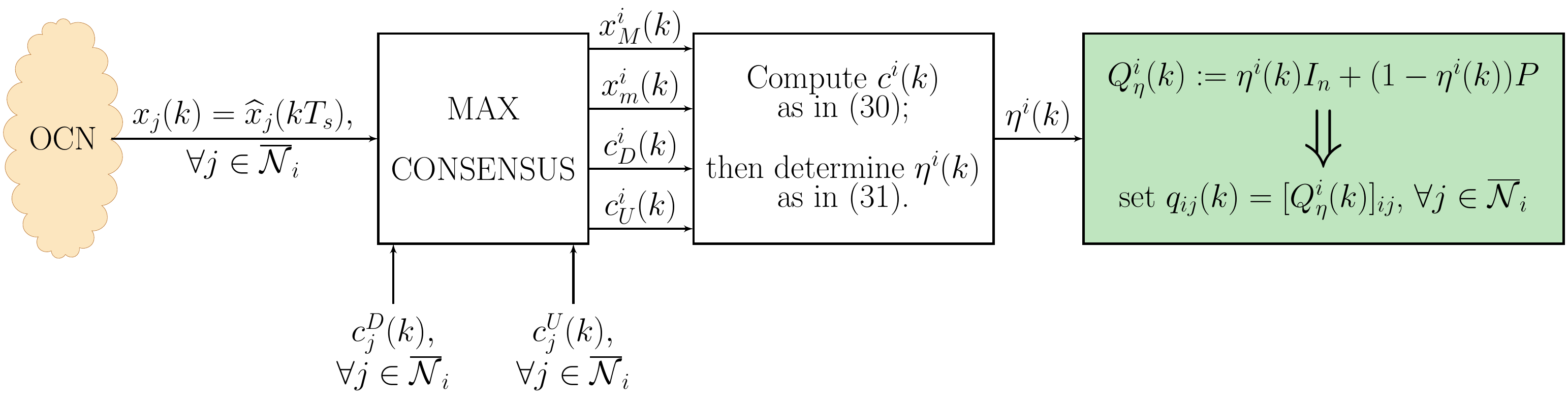}
	\caption{Workflow diagram of the proposed RGP based on Alg.~\ref{alg:1}, seen from the local perspective of the $i$-th channel. With reference to the whole view given in Fig.~\ref{fig:ControlSchemeCavallinoDraft}, this element represents the rightmost part, where the green blocks in the two figures correspond. 
	}
	\label{fig:RGPblockdiagram}
\end{figure*}

\begin{figure*}[t!] 
	\centering
	\vspace{-12pt}
	\subfloat[Given graph $\Gmc^{o}$]{\includegraphics[scale=0.28, clip ]{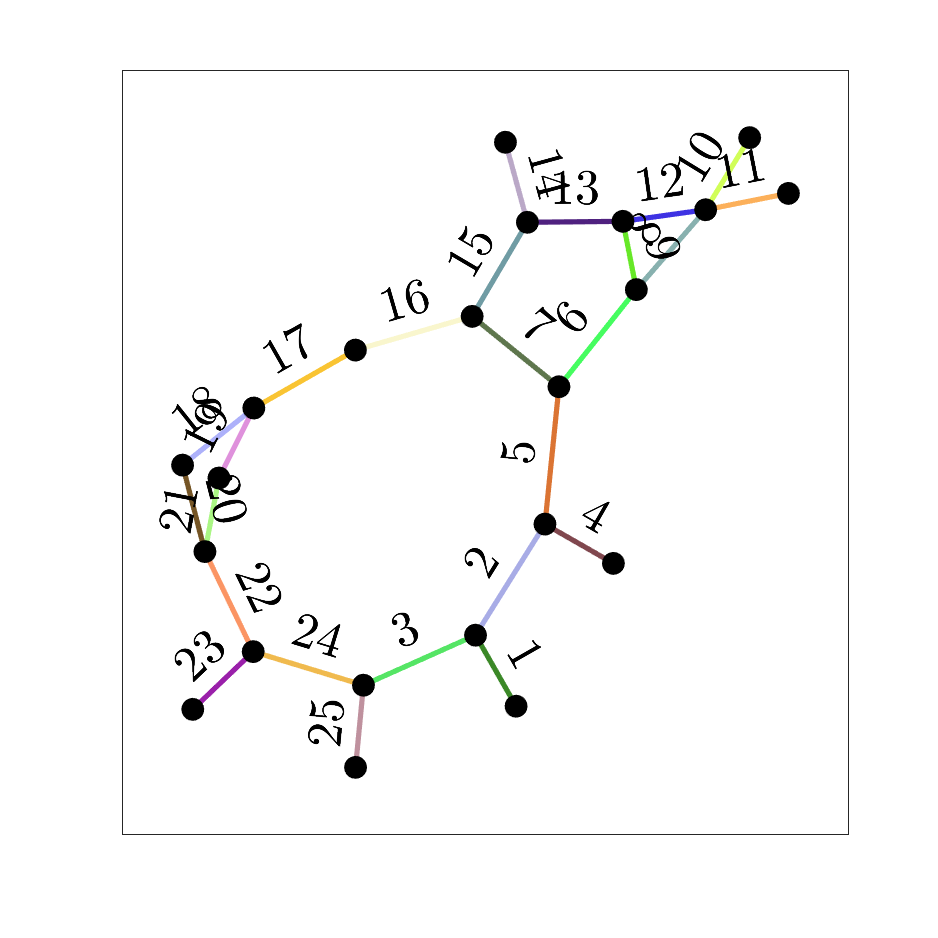}\label{fig:Go1}} 
	\subfloat[Adjoint graph $\Gmc = \mathrm{L}(\Gmc^{o})$]{\includegraphics[scale=0.28, clip]{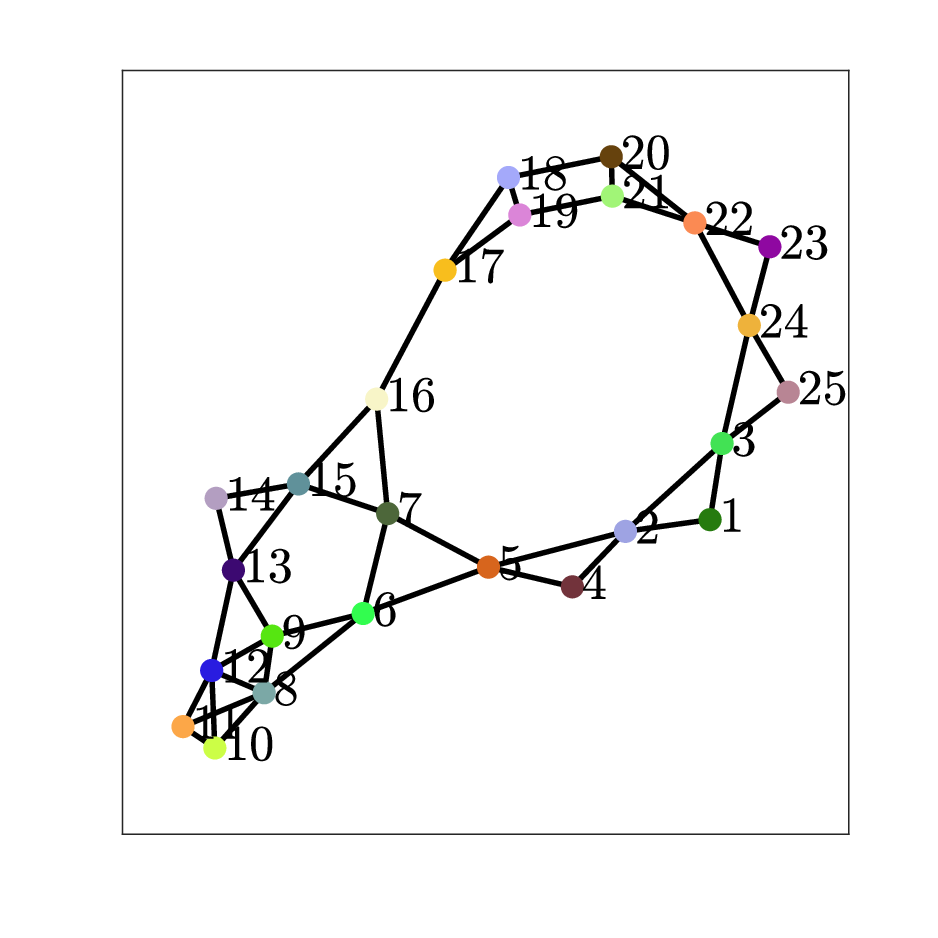}\label{fig:G1}}
	\subfloat[Given graph $\Kmc^{o}$]{\includegraphics[scale=0.28, clip ]{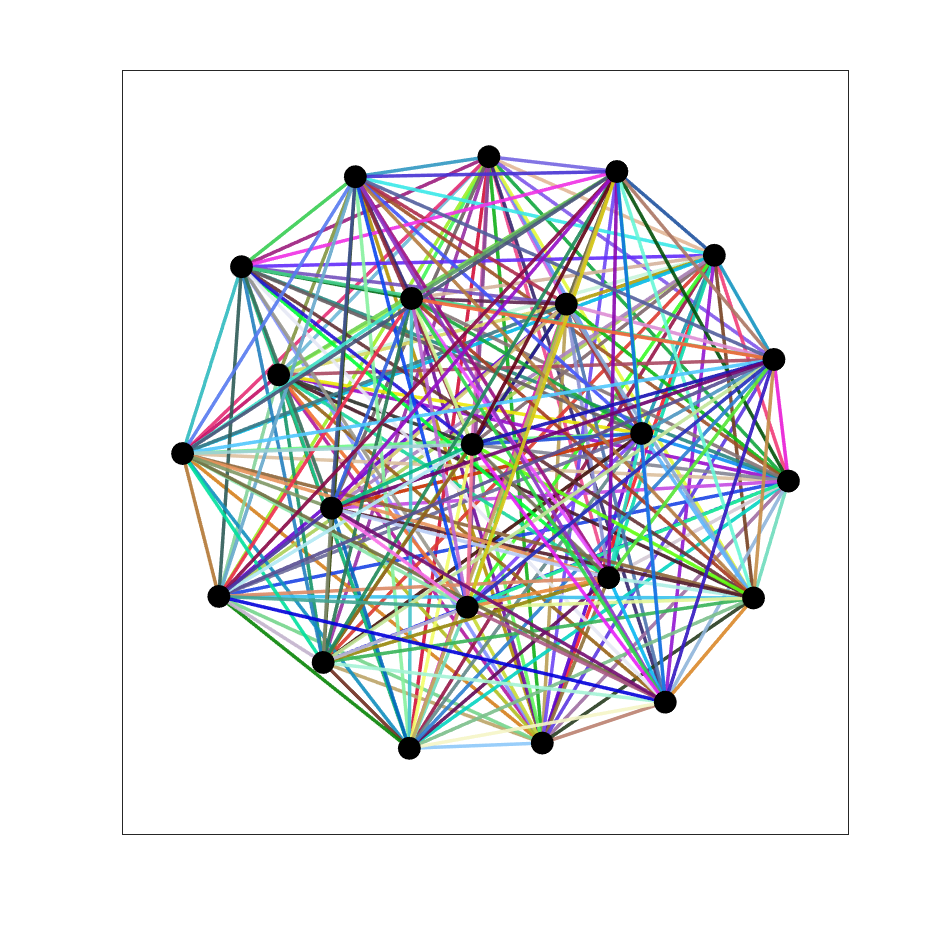}\label{fig:Go2}}
	\subfloat[Adjoint graph $\Tmc = \mathrm{L}(\Kmc^{o})$]{\includegraphics[scale=0.28, clip ]{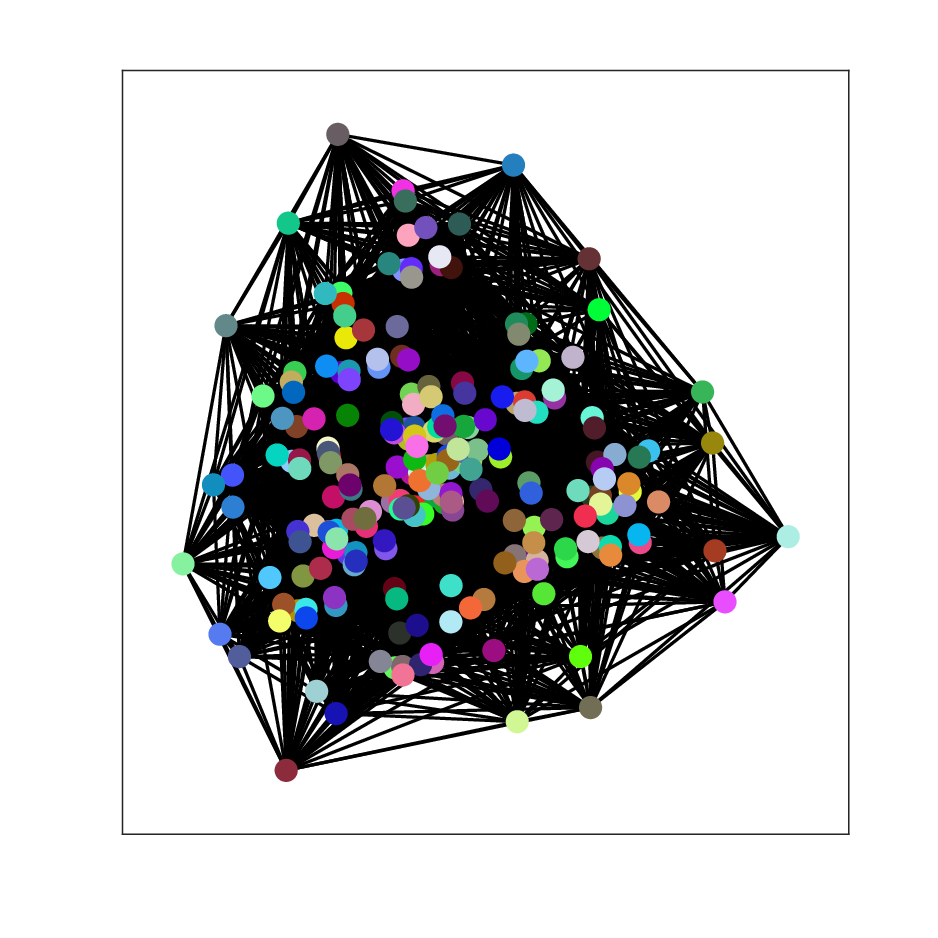}\label{fig:G2}}
	\caption{(a)-(b): graphs capturing the topological information of the Cavallino network of Fig. \ref{fig:cavallinowaternetwork}. (c)-(d): complete-graph version of the Cavallino network for comparison purposes.}
	\label{fig:CavallinoTopology}
\end{figure*}


\begin{algorithm}[t]
	\caption{Adaptive distributed consensus-based RGP}
	\label{alg:1}
	\begin{algorithmic}[1]
		\Require $\Gmc=(\Vmc,\Emc)$, $P$, $\phi^{i}$, $\omega^{i}$, $\eta_{L}^{i}$, $\gamma$, $x_{i}(0)$, $\forall i=1,\ldots,n$;
		\Require $c_{i}^{D}(k)$, $c_{i}^{U}(k)$, $\forall i=1,\ldots,n$ and $\forall k=0,1,2,\ldots$
		\Ensure agreement conditions \eqref{eq:terminationcond1}-\eqref{eq:terminationcond2}
		\State $k\gets0$\label{startAlg1}
		\While{True} 
		\For{all $i =1,\ldots,n$}\label{b:forMCP}
		\State $x_{M}^{i}(k) = MCP_{\Gmc}(x_{j}(k), ~\forall j \in \overline{\Nmc}_{i})$ \label{b:MC1}
		\State $x_{m}^{i}(k) = -MCP_{\Gmc}(-x_{j}(k), ~\forall j \in \overline{\Nmc}_{i})$
		\State $c^{i}_{D}(k) = -MCP_{\Gmc}(-c_{j}^{D}(k), ~\forall j \in \overline{\Nmc}_{i})$
		\State $c^{i}_{U}(k) = -MCP_{\Gmc}(-c_{j}^{U}(k), ~\forall j \in \overline{\Nmc}_{i} )$ \label{e:MC1}
		\EndFor\label{e:forMCP}
		\For{all $i =1,\ldots,n$}\label{b:forxk1}
		\If{$x_{M}^{i}(k) - x_{m}^{i}(k) \leq \gamma$}\label{b:ifW}
		\State \textbf{break} the main \textbf{while} loop
		\EndIf\label{e:ifW}
		\State $c^{i}(k) = \min(c^{i}_{D}(k),c^{i}_{U}(k))$\label{b:update}
		\State $\eta^{i}(k) =$
		\NoNumber{$ \max(\eta_{L}^{i}, 1-c^{i}(k)/ (\omega^{i} \max(-x_{m}^{i}(k),x_{M}^{i}(k)) ))$}  
		\State $x_{i}(k+1) =  $\label{e:update}
		\NoNumber{$\eta^{i}(k)x_{i}(k)+(1-\eta^{i}(k)) \sum_{j \in \overline{\Nmc}_{i} } p_{ij} x_{j}(k)$}
		\EndFor\label{e:forxk1}
		\State $k \gets k+1$
		\EndWhile
	\end{algorithmic}
\end{algorithm}


Specifically, the proposed algorithm (pseudocode in Alg. \ref{alg:1}) requires the following information: the (detrended) initial condition $x(0)$, the adjoint graph $\Gmc$ of the given network topology $\Gmc^{o}$, the constraint functions $c_{i}^{D}$ and $c_{i}^{U}$, for all $i=1,\ldots,n$. Moreover, each agent $i$ needs to store beforehand a local copy -- indicated by the same symbol with superscript $i$ -- of constant $\eta_L$ computed as in \eqref{etal}, weights $p_{ij}$, $\forall j \in \Nmc_{i}$, obtained via \eqref{metro} in relation to $\Gmc$, threshold $\gamma$ and the diameter $\phi$ of the underlying network $\Gmc$. It is worth noting that quantities such as $\eta_{L}$ or $\phi$ could be also retrieved during the initialization of this procedure in a distributed fashion, exploiting decentralized algorithms (see e.g. ~\cite{ThiAlain2014,GusrialdiQu2017} to estimate the eigenvalues needed in the calculations of $\eta_{L}$ and ~\cite{PelegRodittyTal} for $\phi$).

From the first half of Alg. \ref{alg:1} (lines \ref{startAlg1}--\ref{e:forMCP}), it is possible to understand how preliminary quantities are retrieved in a distributed way. Indeed, lines \ref{b:forMCP}-\ref{e:forMCP} represent an intermediate call to a subroutine running the MCP. 
In particular, the computation of $x_{M}^{i}(k)$ and $x_{m}^{i}(k)$ allows to obtain a local version of quantities $\left\|x(k)\right\|_{\infty} = \max(|x_{m}^{i}(k)|,|x_{M}^{i}(k)|) $, $\forall i = 1,\ldots,n$, and $W(x(k))=x_{M}^{i}(k)-x_{m}^{i}(k)$, $\forall i = 1,\ldots,n$. Remarkably, the four different invocations at lines \ref{b:MC1}-\ref{e:MC1} to the 
MCP can be parallelized. To be precise, quantities $c^{i}_{D}(k)$ and $c^{i}_{U}(k)$ could be determined beforehand, as these do not depend on the state $x(k)$. 
We also recall that the MCP converges over any given undirected and connected topology $\Gmc$ and returns the maximum value of the considered quantity 
in at most a time proportional to the diameter $\phi$. This fact is exploited to guarantee that this stage of the main Alg.~\ref{alg:1} be executed within the time interval $[k,k+1]$.

The second part of Alg.~\ref{alg:1} from line \ref{b:forxk1} to line \ref{e:forxk1} addresses the main body of the provided solution to the reference generation. In particular, lines \ref{b:ifW}-\ref{e:ifW} implement the termination condition described in \eqref{eq:terminationcond1} (exiting the {\bf while} loop), and lines \ref{b:update}-\ref{e:update} illustrate how the state $x(k)$ is locally updated into $x(k+1)$ according to equations \eqref{eq:3}, \eqref{eq:c_min} and \eqref{eq:3.33}.
As a final note, we highlight how the two {\bf for} loops at lines~\ref{b:forMCP} and \ref{b:forxk1} are indeed parallel executions on all nodes and not sequential operations. An overview on the fundamental mechanism beneath the proposed RGP can be now sketched in the workflow diagram depicted in Fig. \ref{fig:RGPblockdiagram}.

To conclude, we point out that a measure of the convergence rate of Alg.~\ref{alg:1} can be given by  

\begin{equation}\label{eq:Rconvergence}
	R :=  \phi \left(1+\dfrac{d_{M}-d_{m}}{2}\right)^{\rho} \geq 1.
\end{equation}
Specifically, the lower the value of $R$ the faster Alg.~\ref{alg:1} terminates satisfying the agreement conditions \eqref{eq:terminationcond1}-\eqref{eq:terminationcond2},
since the expression in \eqref{eq:Rconvergence} is established upon the results obtained from the following theorem (proof and more details in App.~\ref{app:convergenceanalysis}). Moreover, the rate $R$ takes into account the computational burden due to max-consensus protocol, which requires at most $\phi$ steps to be run. The lower bound $R=1$ is attained when $\Gmc$ is \textit{regular} ($d_{M}=d_{m}$) and $\phi=1$: this occurs if and only if an OCN has $m=3$ junctions.

\begin{theorem}\label{thm:convergence}
		The RGP \eqref{eq:3}, with $\eta(k)$ selected as in \eqref{eq:3.33}, converges to average consensus as $k \rightarrow \infty$. In particular, by taking $W(x(k))$ in \eqref{eq:lyap} as a Lyapunov function, $W(x(k))$ vanishes with an approximate exponential decay rate $r \in (0,1)$ depending on the graph radius $\rho$ and the minimum nonzero entry of $Q_{\eta}(k)$ over time $k$. 
		For dynamics in \eqref{eq:3}, discrete-time consensus (Def.~\ref{def:consensus}) is achieved as
		\begin{equation}\label{eq:averageconsensusdef}
			\underset{k\rightarrow\infty}{\lim} x_{i}(k) = \alpha, \quad \forall i = 1,\ldots,n,
		\end{equation}
		with $\alpha$ defined in Asm. \ref{ass} and
		\begin{equation}\label{eq:W0converging}
			0 \leq \underset{l\rightarrow\infty}{\lim} W(l \rho) \leq \underset{l\rightarrow\infty}{\lim} r^{l} W(x(0)) =0,
		\end{equation}
		with $l\in \mathbb{N}$ and being $r$ upper bounded by
		\begin{equation}\label{eq:upperboundr}
			\overline{r} = 1-\left((1-\eta_{H})\underline{\xi} \right)^{\rho} < 1,
		\end{equation}
		and lower bounded by
		\begin{equation}\label{eq:lowerboundr}
			\underline{r} = 1- \left(\overline{\xi}+(1-\overline{\xi})\eta_{H}\right)^{\rho} > 0,
		\end{equation}
		where the topology-dependent  constants $\underline{\xi}$, $\overline{\xi}$ and the upper bound $\eta_{H}$ are respectively defined as in \eqref{eq:minentryP}, \eqref{eq:maxentryP}, and \eqref{eq:etaHdef}.
	\end{theorem}


\section{Numerical simulations} \label{sec:numericalsimulations}

\begin{table*}[t!]
	\centering
	\begin{tabular}{|c| c| c| c| c| c| c| c| c| c| c| c|}\hline 
		$    $ & $\varsigma_{P}$ & $d_{m}$ & $d_{M}$ & $\omega$ & $\eta_{L}$ & $\eta_{H}$ & $\underline{\xi}$ & $\overline{\xi}$ & $\rho$ & $\phi$ 
		& $R$  \\
		\hline
		$\Gmc$ & $0.371$ & $2$ & $5$ & $1.667$ & $\zeta$ & $0.912$ & $0.167$ & $0.667$ & $5$ & $7$ 
		& $683.6$ \\\hline    
		$\Tmc$ & $0.220$ & $40$ & $40$ & $1.951$ & $\zeta$ & $0.925$ & $0.024$ & $0.024$ & $2$ & $2$ 
		& $2$  \\\hline
	\end{tabular}
	\caption{Topological constants pertaining to graphs $\Gmc$ and $\Tmc$.}
	\label{Tab:topopar}
\end{table*}

\begin{figure*}[t!] 
	\centering
	\vspace{-6pt}
	\subfloat[Average consensus dynamics]{\includegraphics[width=0.25\textwidth, 
		]{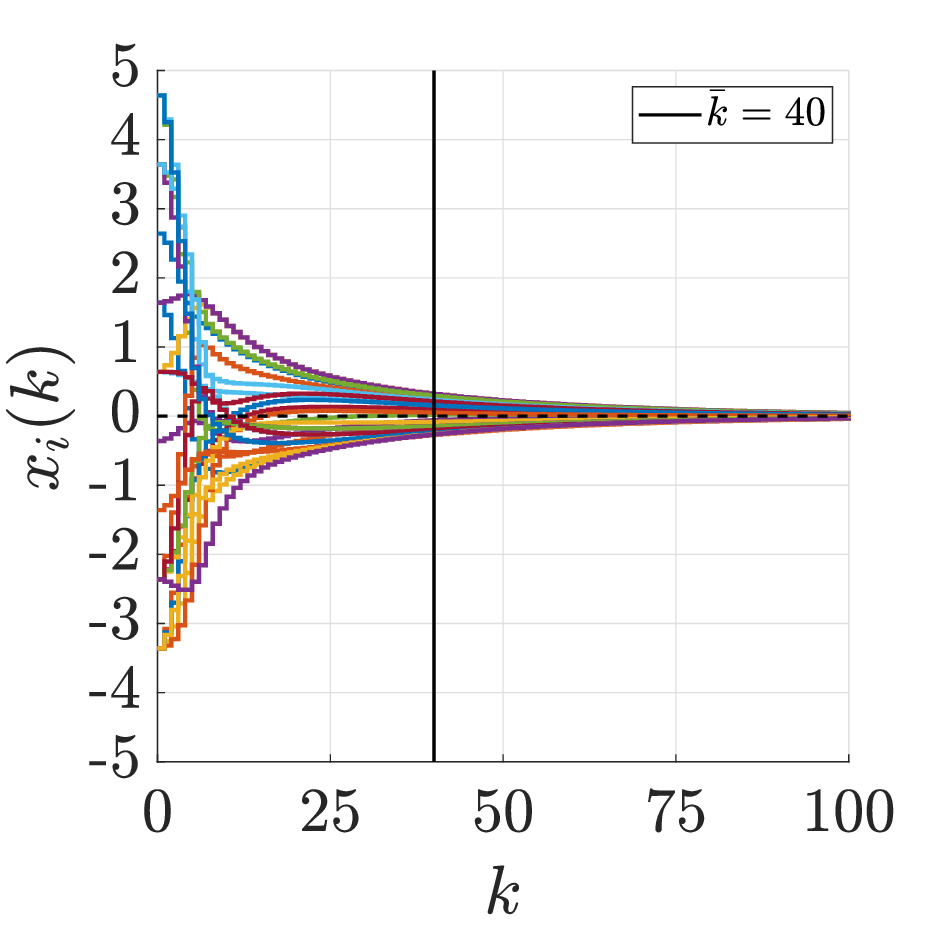}\label{fig:traj1}}
	\subfloat[Disagreement function decay]{\includegraphics[width=0.25\textwidth, 
		]{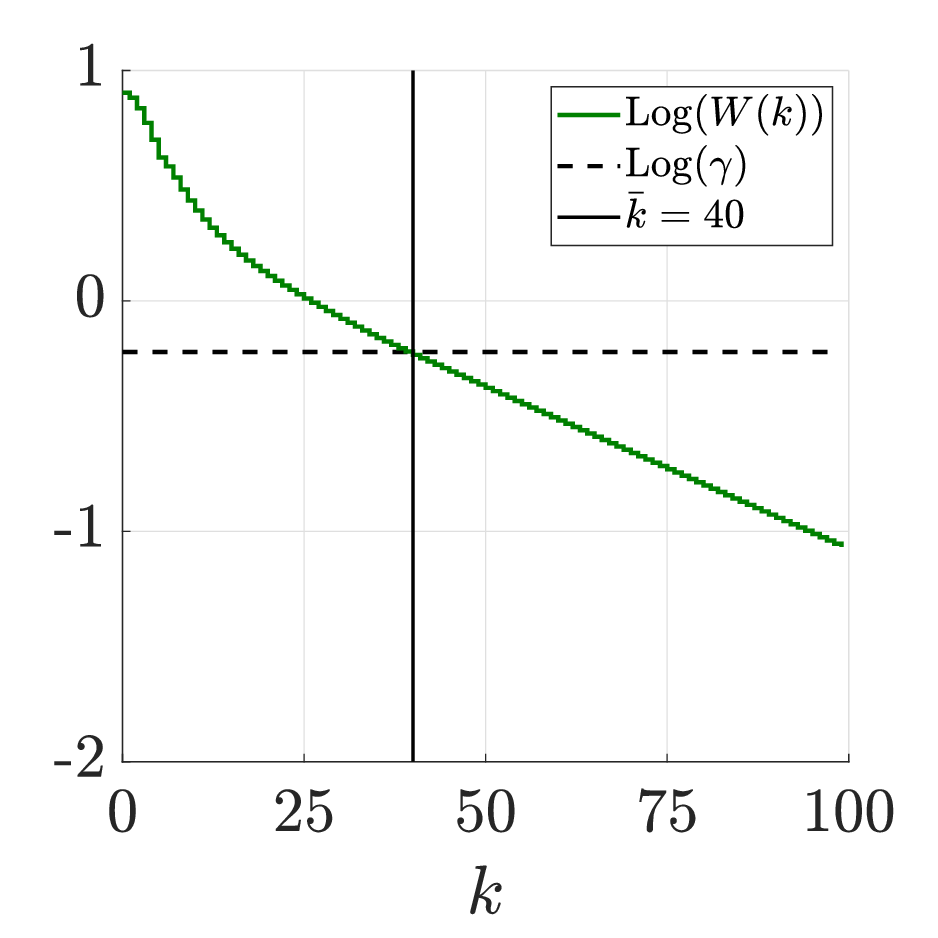}\label{fig:W1}}
	\subfloat[Feasibility check]{\includegraphics[width=0.25\textwidth, 
		]{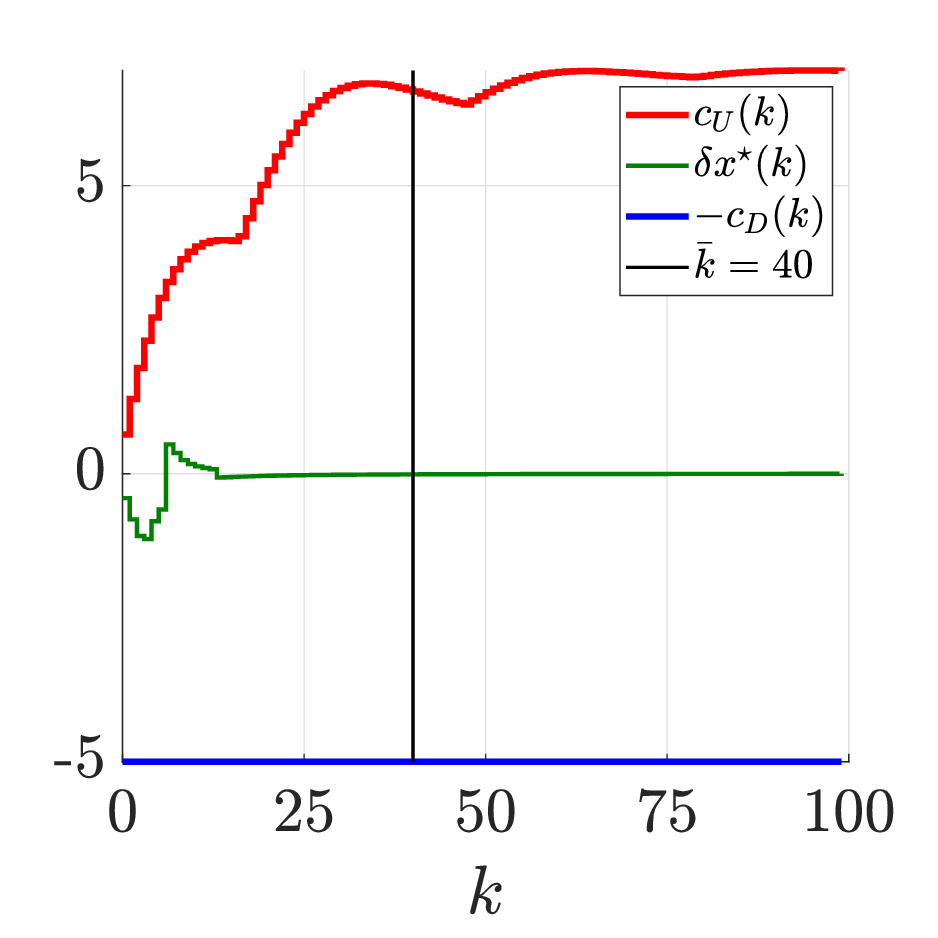}\label{fig:constr1}}
	\subfloat[Behavior of $\eta(k)$]{\includegraphics[width=0.25\textwidth, 
		]{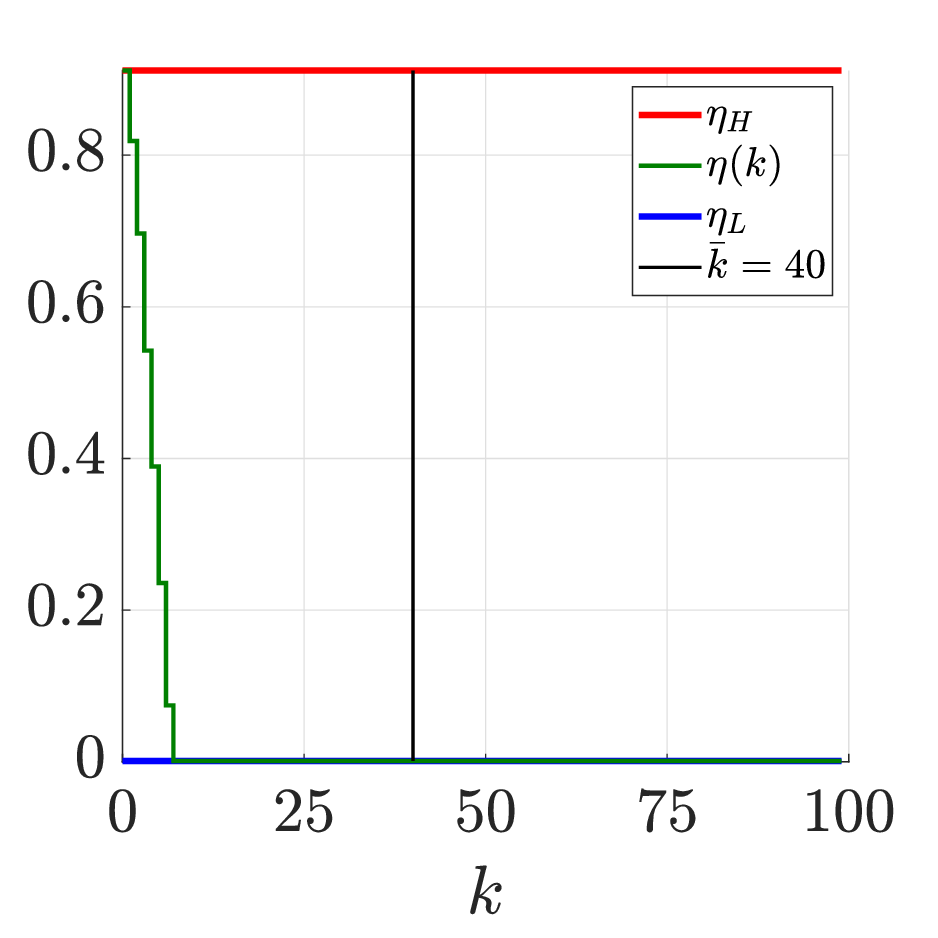}\label{fig:eta1}}\\
	\subfloat[Average consensus dynamics]{\includegraphics[width=0.25\textwidth, 
		]{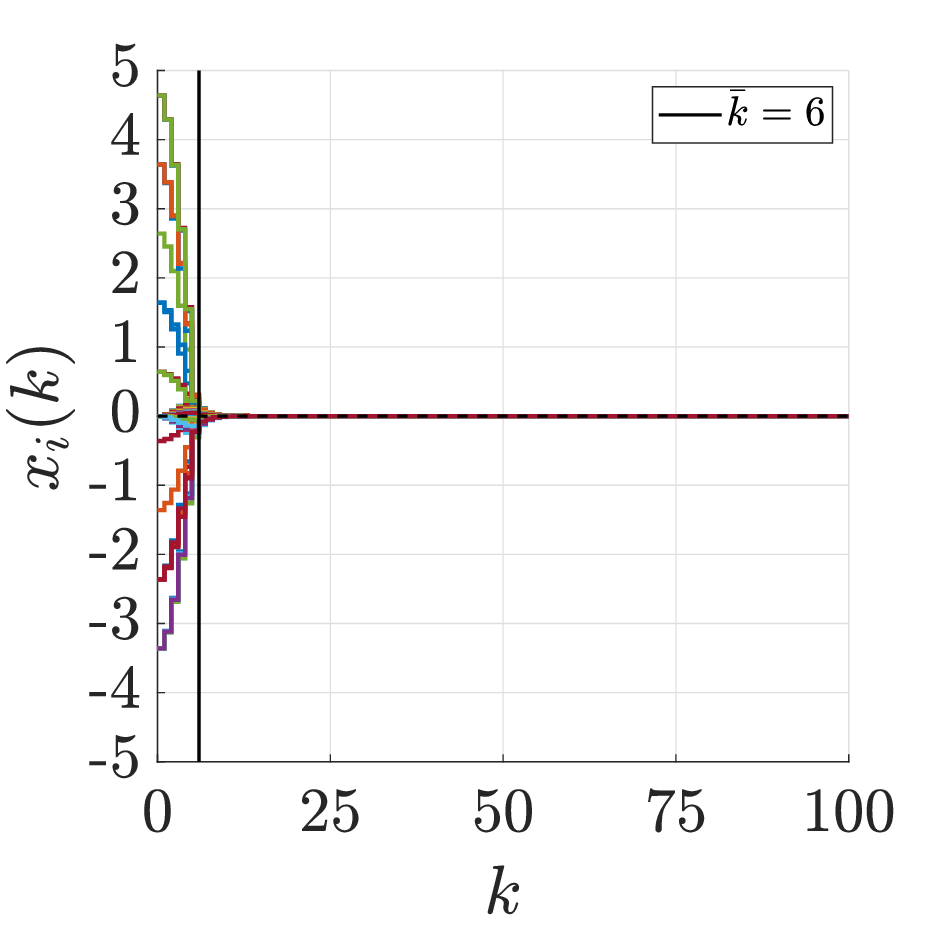}\label{fig:traj2}}
	\subfloat[Disagreement function decay]{\includegraphics[width=0.25\textwidth, 
		]{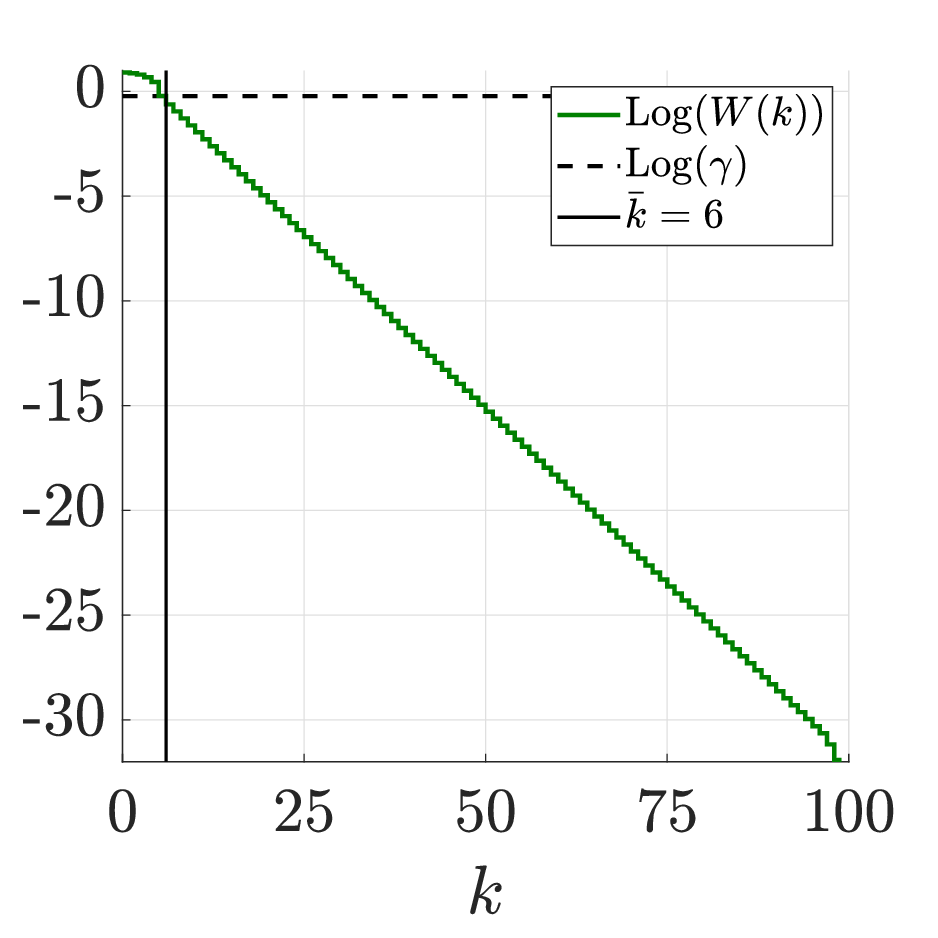}\label{fig:W2}}
	\subfloat[Feasibility check]{\includegraphics[width=0.25\textwidth, 
		]{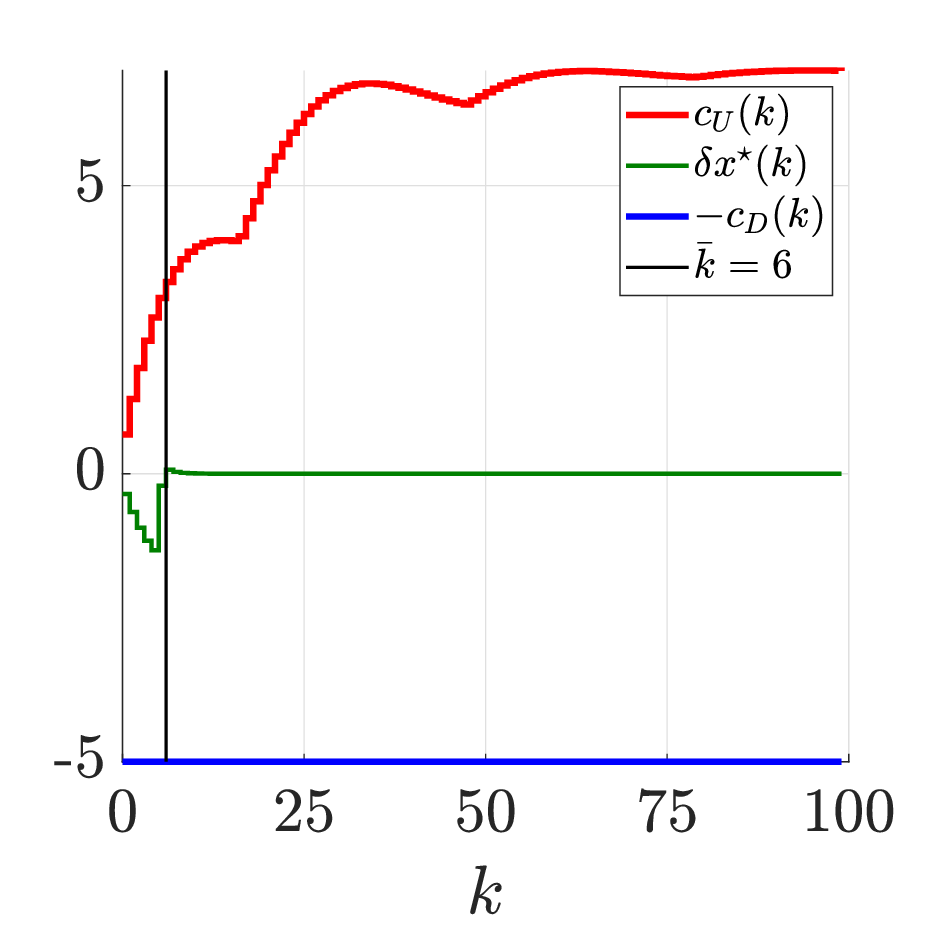}\label{fig:constr2}}
	\subfloat[Behavior of $\eta(k)$]{\includegraphics[width=0.25\textwidth, 
		]{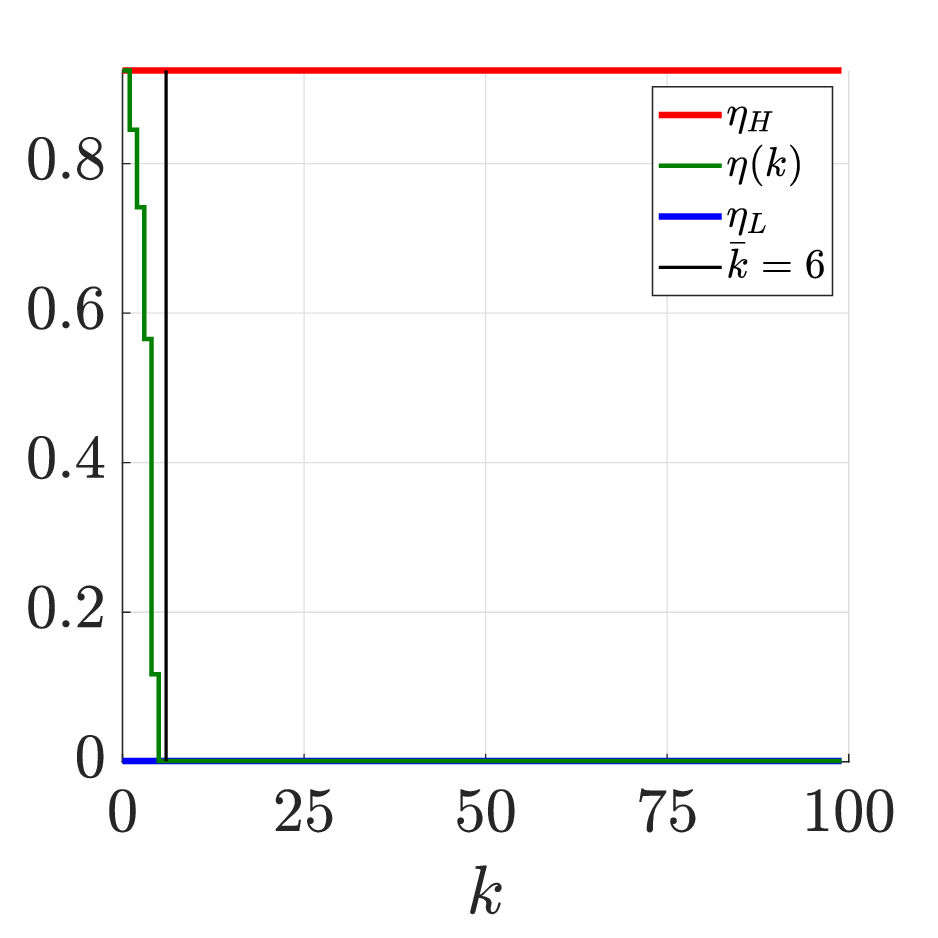}\label{fig:eta2}}\\
	\caption{Reference generation (Alg.~\ref{alg:1}) for the Cavallino OCN case study applied on graphs $\Gmc$ (a)-(d) and $\Tmc$ (e)-(h). Panels (a)-(e) show the behavior of the controlled quantity and (b)-(f) reports the related disagreement function. In (c)-(g) it appears that the level variation $\delta x^\ast$ does not violate the given channel constraints.}
	\label{fig:sim_Cavallino}
\end{figure*}

We discuss here some numerical simulations to support our previous results. 
We consider a portion of the Cavallino network in Fig. \ref{fig:cavallinowaternetwork} having $m=22$ junctions and $n_{1}=25$ channels, and compare it with its corresponding \textit{complete} topology (see Fig.~\ref{fig:CavallinoTopology}). This latter is characterized by the same number of junctions ($m=22$) and the maximum allowed number of channels $n_{2} = m(m-1)/2= 231$: in this sense, although it may result unfeasible from a practical point of view, it represents an ideal fully connected reference structure. 
In particular, we denote with $\Gmc^{o}$ the graph modeling the Cavallino network, illustrated in Fig. \ref{fig:Go1}, and with $\Kmc^{o}$ its complete version, depicted in Fig. \ref{fig:Go2}. 
The adjoint graph $\Gmc = \mathrm{L}(\Gmc^{o})$ of $\Gmc^{o}$ with $n_{1}$ nodes is obtained as in Fig. \ref{fig:G1}, while the adjoint graph $ \Tmc = \mathrm{L}(\Kmc^{o})$ of $\Kmc^{o}$ is the \textit{triangular} graph~\cite{BrualdiRyser1991} with $n_{2}$ nodes, represented in Fig.~\ref{fig:G2}.
For both networks $\Gmc$ and $\Tmc$, assumptions Asm.~\ref{asm:network} - Asm.~\ref{ass} are made and the following parameters are set: $T_{s}=1$, $\gamma = 0.6$, $k_{MAX} = 100$, $\zeta = 0.001$. The constraint functions are uniformly chosen for all $i$ as $c_{i}^{D}(k) = 5$, $\forall k \in \mathbb{N}$, and $c_{i}^{U}(k) = 7(1-0.95^{k+1}) [1- 0.95^{k+1} |\cos(k/10)|] \geq 0.6825$. 
For both $\Gmc$ and $\Tmc$, the corresponding initial conditions $x^{(\Gmc,0)}$ and $x^{(\Tmc,0)}$ satisfy $\left\|x^{(\Gmc,0)}\right\|_{\infty} = \left\|x^{(\Tmc,0)}\right\|_{\infty} \simeq 4.64$, as they are fairly selected to verify

\begin{equation*}
	\forall e_{\ell_{2}} \in \Emc^{o}_{2} :~~ x^{(\Tmc,0)}_{e_{\ell_{2}}} = \begin{cases}
		x^{(\Gmc,0)}_{e_{\ell_{1}}}	, \quad & \text{if } e_{\ell_{2}} \in \Emc_{1}^{o}; \\
		0, \quad & \text{otherwise}.
	\end{cases}
\end{equation*}
%

Within this setup, the remaining topological parameters characterizing graphs $\Gmc$ and $\Tmc$ are collected in Tab \ref{Tab:topopar}.  

\begin{figure*}[t!] 
	\centering
	\vspace{-6pt}
	\subfloat[Average consensus dynamics]{\includegraphics[width=0.25\textwidth, 
		]{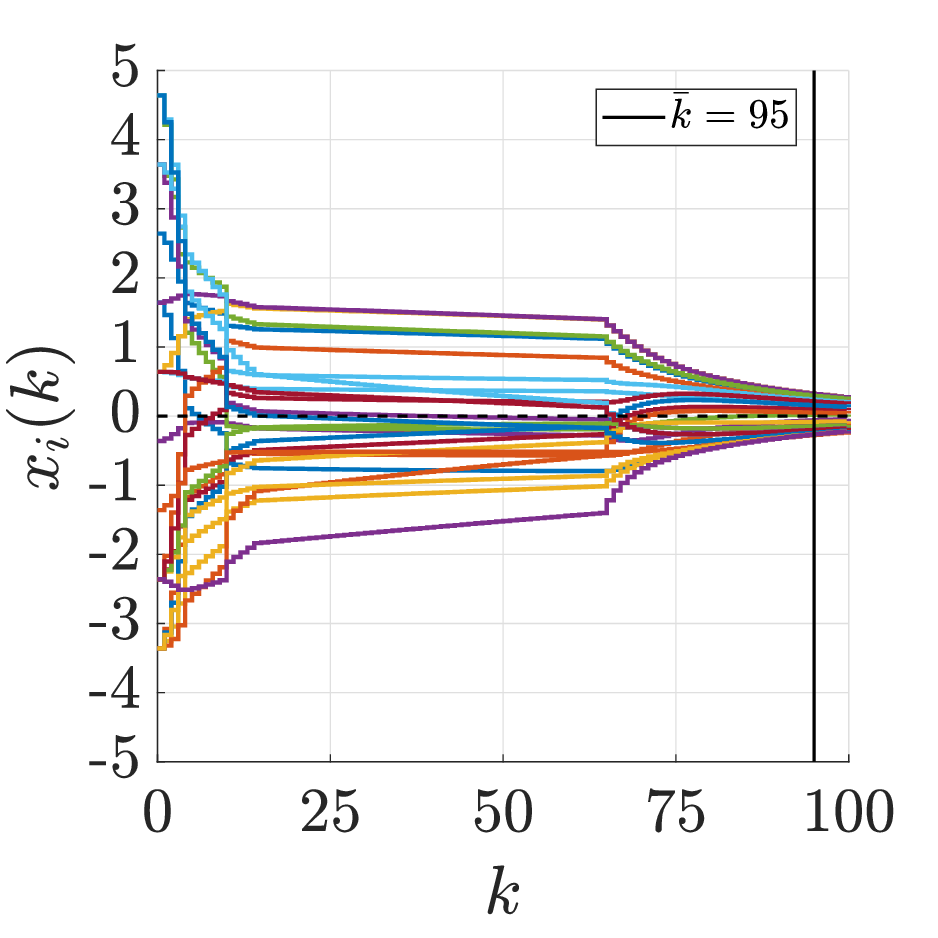}\label{fig:traj1f}}
	\subfloat[Disagreement function decay]{\includegraphics[width=0.25\textwidth, 
		]{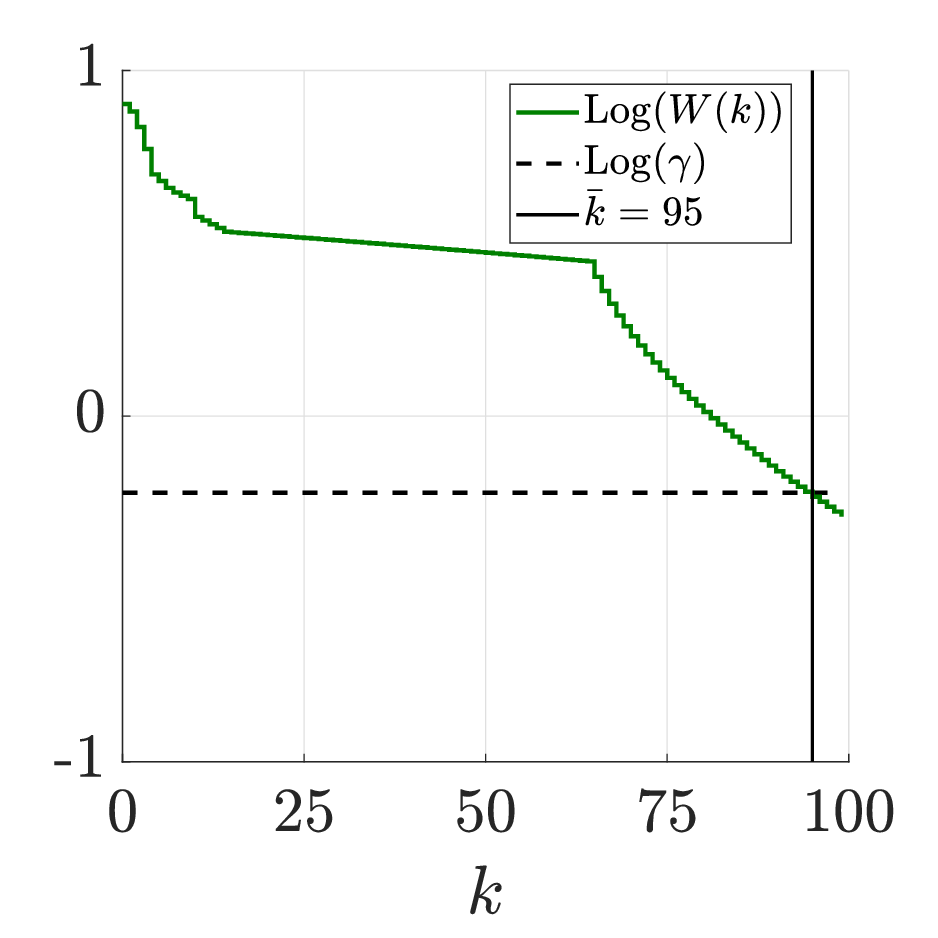}\label{fig:W1f}}
	\subfloat[Feasibility check]{\includegraphics[width=0.25\textwidth, 
		]{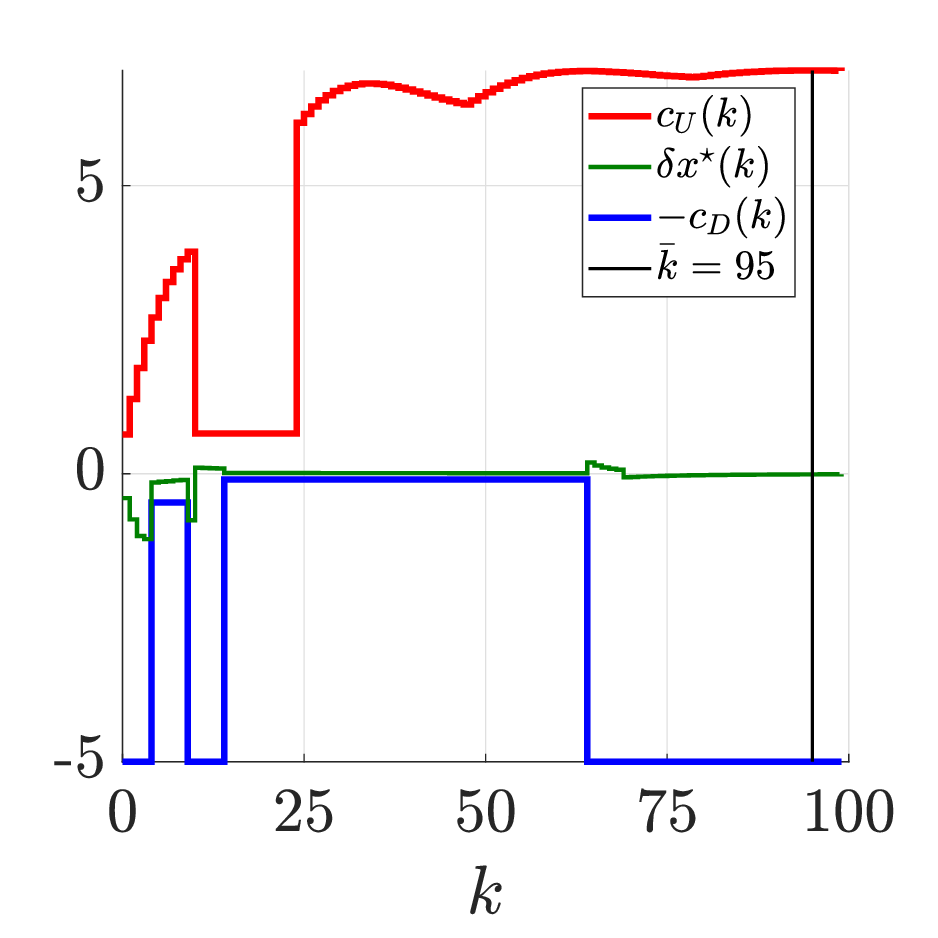}\label{fig:constr1f}}
	\subfloat[Behavior of $\eta(k)$]{\includegraphics[width=0.25\textwidth, 
		]{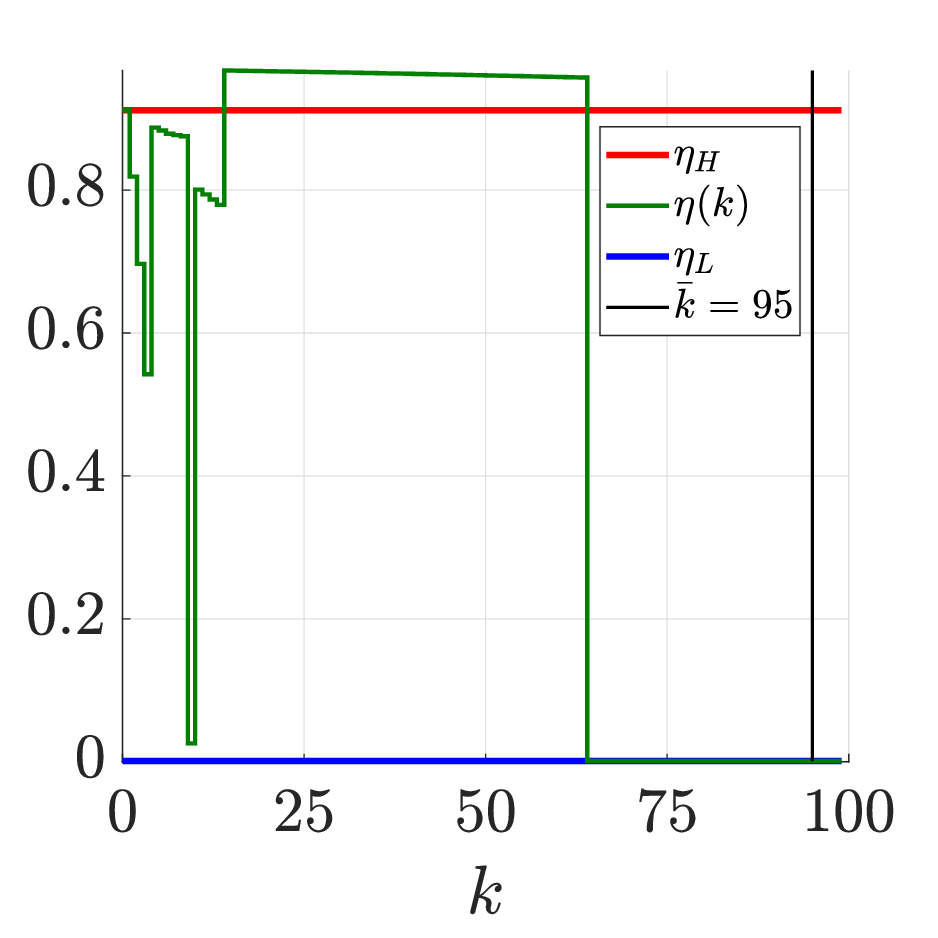}\label{fig:eta1f}}
 \caption{Reference generation (Alg.~\ref{alg:1}) for the Cavallino OCN case study applied on graphs $\Gmc$ for a faulty scenario where the channel constraint capacity values are artificially abruptly modified ($c_U$,$c_D$ in (c)).}
	\label{fig:sim_Cavallino2}
\end{figure*}

Fig. \ref{fig:sim_Cavallino} depicts the outcome of Alg. \ref{alg:1} executed on networks $\Gmc$ and $\Tmc$, respectively. 
Specifically, Figs. \ref{fig:traj1}-\ref{fig:traj2} show how the state trajectories representing the water increment reference evolve as a function of time, starting from a generic initial condition: in practice, starting from an uneven water distribution over the network, Alg. \ref{alg:1} guarantees to reach an equilibrium among the channels.

It is worth to note that the dynamics corresponding to $\Tmc$ converges to average consensus faster than that of $\Gmc$, as $\Tmc$ is a regular graph with higher degree $d_{M2}=d_{m2}=40$ w.r.t. to those of $\Gmc$ (ranging from $2$ to $5$); furthermore, $\Tmc$ has a total of $(n_{2}d_{M2}/2)=4620$ edges and its diameter is given by $\phi_{2} = \min(2,n_{2}-2)=2$. 
These two facts imply that information is exchanged far more rapidly among the nodes of $\Tmc$ (see also the convergence rate indexes in Tab. \ref{Tab:topopar}). The faster convergence of dynamics related to $\Tmc$ is more explicitly shown in Fig. \ref{fig:W2}, where for $\bar{k}=6$ it occurs that $W(\bar{k}) := W(x(\bar{k})) \leq \gamma$. Indeed, the max-min disagreement function decreases more quickly w.r.t. that in Fig. \ref{fig:W1}, wherein $\bar{k}=40$ iterations are needed to execute the whole water distribution algorithm and reach a balanced condition. 
Nonetheless, from both these latter plots it is possible to appreciate the exponential decay of $W(k)$. 
In Figs.~\ref{fig:constr1}-\ref{fig:constr2} it is illustrated the behavior of quantity $\delta x^{\star}(k) := S_{j}(k)\left\|x(k)\right\|_{\infty}$, where $S_{j}(k) \in \{-1,0,1\}$ is the sign of the component $x_{j}(k)$ such that $|x_{j}(k)| \geq |x_{i}(k)|$ holds for all $i = 1,\ldots,n$, with $n\in \{n_{1},n_{2}\}$. Clearly, the imposed constraints are not violated for all $k=0,1,2,\ldots$, since $\delta x^{\star}(k)$ remains bounded within functions $-c_{i}^{D}$ and $c_{i}^{U}$.
However, Fig.~\ref{fig:constr2} shows that trajectory $|\delta x^{\star}(k)| = \left\|x(k)\right\|_{\infty}$ exhibits a more peaked and faster behavior than the corresponding one in Fig. \ref{fig:constr1}: this is due to the fact that $\Gmc$ is a less connected structure, hence less subject to fast and strong variations of the dynamics. 
Moreover, it is worth to observe Figs. \ref{fig:eta1}-\ref{fig:eta2}, as they detail the evolution of parameter $\eta(k)$. Counterintuitively, $\eta(k)$ approaches the suboptimal value $\eta_{L}$ faster when Alg.~\ref{alg:1} is implemented on $\Gmc$. In reality, this behavior is a consequence of the fact that dynamics corresponding to $\Tmc$ is faster, and for this reason it needs to be slowed down by higher values of $\eta(k)$ in order not to violate the imposed water flow constraints.

Finally, a further scenario is also considered, in which the upload and download constraint functions are artificially modified to mimic a sudden and unpredicted failure of the channel capacities. 
It can be seen in Figs.~\ref{fig:traj1f}-\ref{fig:eta1f} that, under the application of Alg.~\ref{alg:1} also in such non-nominal conditions, the abrupt $c_U$-$c_D$ variations (Fig.~\ref{fig:constr1f}) yield some real-time adjustment of the $|\delta x^{\star}(k)|$ variable, which is reflected in the overall slower convergence ($\bar{k}=95$) with respect to the no-fail case of Figs.~\ref{fig:traj1}-\ref{fig:eta1}. Nonetheless, also in this situation, convergence is attained despite the violation of the theoretical $\eta$ bounds, as shown in Fig.~\ref{fig:eta1f}. 

In the light of this analysis, it can be observed that the proposed Alg.~\ref{alg:1} operates correctly and robustly over diverse and realistic network conditions, despite that there exist obvious differences dictated by the underlying topology on how fast the agreement is reached. 
Indeed, as requested by Problem~\ref{problem}, it is guaranteed convergence to average consensus for the input reference feeding each agent in Fig.~\ref{fig:ControlSchemeCavallinoDraft} and constraints \eqref{eq:heightconstraints} are satisfied at all time instants $k=0,1,2,\ldots$.


\section{Conclusions and continuing research}\label{chp:conclusion}
In this work, an automated solution to the water level regulation within an OCN is presented and it is designed an adaptive consensus-based algorithm providing the reference to attain an even distribution of water levels.
After modeling and characterizing the OCN with the tools of graph theory, the proposed approach leverages fully decentralized computations and it also considers the presence of water exchange capacity limits, which are embedded in the developed iterative procedure as download and upload constraints for each branch of the system. In addition, it is shown that the devised algorithm converges exponentially fast.

A potential research direction is represented by the analysis and design of the entire control scheme accounting for both reference generation and state regulation.
Finally, a couple of interesting investigation extensions are embodied by the development of the presented scheme exploiting directed topologies and by the improvement of its convergence properties.

\appendix
\section{}

\subsection{Preliminary lemmas}\label{app:prelimlemmas}

\begin{lemma}\label{prop:metropolis}
	Let us consider an undirected and connected graph $\Gmc = (\Vmc, \Emc)$. 
	Let us also define the MH matrix $P \in \mathrm{stoch}^{2}(\mathbb{R}^{n \times n})$ associated to $\Gmc$ as in \eqref{metro}.
	Then the smallest nonzero entry $\underset{p_{ij} >0}{\min} \{p_{ij}\} $ of $P$ is yielded by
	\begin{align}\label{eq:minentryP}
		\underset{p_{ij} >0}{\min} \{p_{ij}\} &= \underset{(i,j) \in \Emc}{\min} \{p_{ij}\} = \underset{i=1,\ldots,n}{\min} \{p_{ii}\} \nonumber\\ 
		&= (1+d_{M})^{-1} =: \underline{\xi} \in (0,1).
	\end{align}
	and the largest entry of $P$ can be upper bounded as
	\begin{equation}\label{eq:maxentryP}
		\max\{p_{ij}\} \leq 1-d_{m}(1+d_{M})^{-1} =: \overline{\xi} \in (0,1),
	\end{equation}
	where equality in \eqref{eq:maxentryP} holds strictly if there exists a node $i$ with $d_{i}=d_{m}$ connected to a node $j$ with $d_{j}=d_{M}$.
\end{lemma}

\begin{proof}
	It is straightforward by property \eqref{metro}.
\end{proof}

\begin{lemma}\label{omeg}
	Let us assign 
	\begin{equation}\label{eq:omegadef}
		\omega :=  2d_M (1+d_M)^{-1}  \in [1,2).
	\end{equation}
	Under the same assumptions of Lem. \ref{prop:metropolis}, for any MH matrix $P \in \mathrm{stoch}^{2}(\mathbb{R}^{n \times n})$ it holds that $||I_{n}-P||_{\infty}\leq \omega$.  
\end{lemma}

\begin{proof}
	Exploiting the Gershgorin's disk theorem \cite{Bell1965} and the structure of $P$ dictated by \eqref{metro}, one has
	\begin{align}
		\left\|I_{n}-P\right\|_{\infty} &\leq \max_{i = 1,\ldots,n} \big{\{}  |1-p_{ii}|+\sum_{\forall j \neq i}|p_{ij}| \big{\}} \nonumber \\
		&= \max_{i= 1,\ldots,n}\{2(1-p_{ii})\} = 2\big{(}1-\min_{i= 1,\ldots,n}\{p_{ii}\}\big{)} , \nonumber
	\end{align}
	where $\underset{i= 1,\ldots,n}{\min} \{p_{ii}\} = (1+d_{M})^{-1}$, as shown in Lem. \ref{prop:metropolis}.
\end{proof}

\subsection{Convergence analysis}\label{app:convergenceanalysis}

The proof of Thm. \ref{thm:convergence} provides convergence guarantees towards average consensus for RGP \eqref{eq:3}, whose distributed implementation is performed through Alg. \ref{alg:1}. The related convergence rate is also discussed subsequently.

\begin{proof}
	Define the product matrices $F(k,\Delta k) = \prod_{\sigma=k}^{k+\Delta k} Q_{\eta}(\sigma)$, $B(k) = F(k,\rho-1)$, with $\Delta k \in \mathbb{N}$. Let us also address the entries of matrices $F(k,\Delta k),B(k),Q_{\eta}(k)$ with $f_{ij}(k,\Delta k) = [F(k,\Delta k)]_{ij}$, $b_{ij}(k) = [B(k)]_{ij}$, $q_{ij}(k) = [Q_{\eta}(k)]_{ij}$, and omit the dependency on $k$ for brevity, except where needed to avoid confusion. 
	
	By means of the RGP~\eqref{eq:3} and noting that $A,B \in \mathrm{stoch}^{2}(\mathbb{R}^{n \times n})$, we observe that
	\begin{align}
		W(x(k+\rho)) &= W(B(k)x(k))\nonumber \\
		&= \underset{i = 1,\ldots,n}{\max} \{\textstyle\sum_{l=1}^{n}b_{il}x_{l}\} - \underset{i = 1,\ldots,n}{\min} \{ \textstyle\sum_{l=1}^{n} b_{il}x_{l}\} \nonumber\\
		&= \max_{i = 1,\ldots,n} \{ \textstyle\sum_{l=1, l\neq j}^{n} b_{il}x_{l} + b_{ij} x_{j}\}  \nonumber\\
		&\quad -\underset{i = 1,\ldots,n}{\min} \{ \textstyle\sum_{l=1, l\neq j}^{n} b_{il}x_{l} + b_{ij} x_{j}\} \nonumber\\
		&\leq \max_{i = 1,\ldots,n} \{ (1-b_{ij}) x_{M} + b_{ij} x_{j}\} \nonumber\\
		&\quad - \underset{i = 1,\ldots,n}{\min} \{ (1-b_{ij}) x_{m} + b_{ij} x_{j}\},  \label{eq:Wineq1}
	\end{align}
	where $x_{M}(k) = \max_{i=1,\ldots,n} x_{i}(k)$ and $x_{m}(k) = \min_{i=1,\ldots,n} x_{i}(k)$. Since index $j$ can be generally chosen in $\{1,\ldots,n\}$ w.l.o.g. for inequality \eqref{eq:Wineq1} to hold, it follows that
	\begin{align}\label{eq:Wineq2}
		W(x(k+\rho)) &\leq \underset{j=1,\ldots,n}{\min} \left\lbrace \max_{i = 1,\ldots,n} \{ (1-b_{ij}) x_{M} + b_{ij} x_{j}\} \right\rbrace \nonumber\\
		&\quad - \underset{j=1,\ldots,n}{\max} \left\lbrace \min_{i = 1,\ldots,n} \{ (1-b_{ij})x_{m} + b_{ij} x_{j}\}  \right\rbrace \nonumber\\
		&= (1-b_{\bar{i}\bar{j}})(x_{M}-x_{m}) + b_{\bar{i}\bar{j}} (x_{\bar{j}}-x_{\bar{j}}) \nonumber \\
		&= (1-b_{\bar{i}\bar{j}}) (x_{M}-x_{m}),
	\end{align}
	where $\bar{i} \in \{1,\ldots,n\}$ and $\bar{j} \in \{1,\ldots,n\}$ are such that $b_{\bar{i}\bar{j}} = \max_{j=1,\ldots,n}\min_{i=1,\ldots,n} \{b_{ij}\}$.
	It is worth noticing that $b_{\bar{i}\bar{j}}(k) = f_{\bar{i}\bar{j}}(k,\rho-1) > 0$ for all $k=0,1,2,\ldots$ because at least $\rho$ factors $Q_{\eta}(k),\ldots,Q_{\eta}(k+\rho-1)$ are needed in order to obtain one positive column in the product matrix $F(k,\Delta k)$. The latter property is indeed ensured by the fact that the underlying graph $\Gmc$ is connected and each vertex in $\Gmc$ has a self-loop $q_{ii}(k) \in (0,1)$ for all $i=1,\ldots,n$ at any given $k$. 
	In particular, by setting
	\begin{equation}\label{eq:epsilondef}
		\epsilon := \underset{k=0,1,2,\ldots}{\inf} \left\lbrace \underset{i =1,\ldots,n}{\min} \left\lbrace  \underset{j \in \overline{\Nmc}_{i}}{\min}~ \left\lbrace q_{ij}(k)\right\rbrace \right\rbrace \right\rbrace ,
	\end{equation}
	one has $q_{ij}(k) \geq \epsilon, q_{ij}(k+1) \geq \epsilon, \ldots, q_{ij}(k+\rho-1) > \epsilon$ and, since each row of $Q_{\eta}(k)$ contains coefficients that perform a convex combination, it holds that
	$f_{\bar{i}\bar{j}}(k,0) \geq \epsilon, f_{\bar{i}\bar{j}}(k,1) \geq \epsilon^{2}, \ldots, f_{\bar{i}\bar{j}}(k,\rho-1) = b_{\bar{i}\bar{j}}(k) \geq \epsilon^{\rho}$ for all $k=0,1,2,\ldots$ 
	
	In this direction, we show that $\epsilon \in (0,1)$ leveraging Lem.~\ref{prop:metropolis} and, in particular, the connectedness of the given graph.\\
	Firstly, we find a lower bound $\underline{\epsilon} \in (0,1)$ for expression \eqref{eq:epsilondef}, such that $\epsilon \geq \underline{\epsilon}$. Recalling Prop. \ref{prop:etabounds}, assume that $\eta(k)$ ranges over the interval $ [\underline{\eta},\overline{\eta}] \subseteq [\eta_{L},\eta_{H}]$ as $k$ varies. Then for the diagonal elements of $Q_{\eta}$ one has $q_{ii} = \eta+(1-\eta) p_{ii} \geq \underline{\eta}+(1-\underline{\eta}) \underline{\xi}$, where $\underline{\xi}$ is defined in \eqref{eq:minentryP}. Whereas, for the off-diagonal elements of $Q_{\eta}$ one has $q_{ij} = (1-\eta)p_{ij} \geq (1-\overline{\eta})\underline{\xi}$, with $i \neq j$. As a consequence, it is obtained $\underline{\epsilon} := (1-\eta_{H})\underline{\xi} \leq (1-\overline{\eta})\underline{\xi} = \min((1-\overline{\eta})\underline{\xi}, \underline{\eta}+(1-\underline{\eta}) \underline{\xi})$, where $\underline{\epsilon} > 0$ and $\eta_{H} \in (0,1)$ is defined as in \eqref{eq:etaHdef}. \\
	Similarly, we find an upper bound $\overline{\epsilon} \in (0,1)$ for expression \eqref{eq:epsilondef}, such that $\epsilon \leq \overline{\epsilon}$. Then for the diagonal elements of $Q_{\eta}$ one has $q_{ii} = \eta+(1-\eta) p_{ii} \leq \overline{\eta}+(1-\overline{\eta}) \overline{\xi}$, where $\overline{\xi}$ is defined in \eqref{eq:maxentryP}. For the off-diagonal elements of $Q_{\eta}$ one has $q_{ij} = (1-\eta)p_{ij} \leq (1-\underline{\eta})\overline{\xi}$, with $i \neq j$. As a consequence, it is obtained $\overline{\epsilon} := \overline{\xi}+(1-\overline{\xi})\eta_{H} = \eta_{H}+(1-\eta_{H})\overline{\xi} \geq \overline{\eta} + (1-\overline{\eta})\overline{\xi} = \max((1-\underline{\eta})\overline{\xi}, \overline{\eta}+(1-\overline{\eta}) \overline{\xi})$, where $\underline{\epsilon} < \overline{\epsilon} < 1$.
	
	At the light of the previous computations, inequality \eqref{eq:Wineq2} can be rewritten as
	\begin{align}\label{eq:Wineq3}
		W(x(k+\rho)) &\leq r W(x(k)), 
	\end{align}
	where the scalar 
	\begin{equation}\label{eq:rdef}
		r = 1-\epsilon^{\rho} \in (0,1)
	\end{equation}
	can be related to the convergence of the RGP dynamics~\eqref{eq:3}, namely, the lower the value of $r$, the faster the convergence of RGP towards average consensus. 
	Indeed, inequality \eqref{eq:Wineq3} guarantees that the agreement for RGP \eqref{eq:3} is fulfilled with exponential decay dictated by \eqref{eq:rdef}: choosing a reference instant $k$, $\forall k \in \mathbb{N}$, it holds that 
	\begin{equation}\label{eq:consensusshown}
		0 \leq \underset{l\rightarrow \infty}{\lim} W(x(k+l \rho)) \leq \underset{l\rightarrow \infty}{\lim} r^{l} W(x(k)) = 0,
	\end{equation}
	with $l\in \mathbb{N}$.
	The latter inequality leads to \eqref{eq:W0converging} by choosing $k=0$: namely, we are observing the system at time instants separated by $\rho$ sampling periods.
	
	Now, we show that the consensus for the RGP \eqref{eq:3} is reached in terms of arithmetic mean of the initial conditions, i.e. the final value $x_{\infty}$ taken by $x(k)$, as $k$ goes to infinity, has the form
	\begin{equation}\label{eq:consaverexpl}
		x_{\infty}: =  \underset{\Delta k\rightarrow\infty}{\lim} F(0,\Delta k) x(0) = \alpha \ones_{n}.
	\end{equation}
	Let us define the matrix $M := \underset{\Delta k\rightarrow\infty}{\lim} F(0,\Delta k)$. 
	Since the agreement is guaranteed to be reached by virtue of~\eqref{eq:consensusshown} and because any $F(k,\Delta k) \in \mathrm{stoch}^{2}(\mathbb{R}^{n \times n})$, for all $ k, \Delta k$, then $M$ exists finite, is doubly-stochastic and its form is yielded by $M=\begin{bmatrix} m_1 \ones_n & \cdots & m_n \ones_n \end{bmatrix}$,
	where $\sum_{j=1}^{n} m_{j} = 1$ and $\sum_{i=1}^{n} m_{j} = 1$, for all $j=1,\ldots,n$. 
	From the structure of $M$, it immediately follows that $m_{j}=n^{-1}$, for all $j=1,\ldots,n$. Therefore, $M$ can be decomposed as $M = n^{-1} \ones_{n} \ones_{n}^{\top} = \alpha \ones_{n} $ leading to the fact that \eqref{eq:consaverexpl} can be rewritten as $x_{\infty} = Mx(0)$ or, equivalently, as in \eqref{eq:averageconsensusdef}.
	
	Finally, upper and lower bounds $\overline{r}$ and $\underline{r}$ in \eqref{eq:upperboundr} and \eqref{eq:lowerboundr} are derived by assigning $\overline{r} := 1-\underline{\epsilon}^{\rho}$ and $\underline{r} := 1-\overline{\epsilon}^{\rho}$.
\end{proof}

The following remarks conclude the discussion on the convergence analysis.
\begin{remark}
	The approximate convergence rate $r$ considered in Thm. \ref{thm:convergence} can be estimated by taking $\hat{r} \in [\underline{r},\overline{r}] \subseteq (0,1)$, e.g. by setting
	\begin{equation}\label{eq:rhat}
		\hat{r} := 1-(2+d_{M}-d_{m})^{-\rho} \in [0.5,1).
	\end{equation}
	Expression \eqref{eq:rhat} is obtained after averaging $\underline{\epsilon}$ and $\overline{\epsilon}$ that appear in the proof of Thm.~\ref{thm:convergence} through the convex combination $(1-\beta)\underline{\epsilon} + \beta \overline{\epsilon}$, choosing $\beta := (2+d_{M}-d_{m})^{-1}$.
	Notice that $\hat{r} \geq 0.5$, as $\Gmc$ is connected with at least $n\geq 2$ channels, implying that 
	$\rho \geq 1$.
	Thus, being an estimation of $r$, $\hat{r}$ serves as a purely topological index that can roughly measure how quickly this convergence takes place. Alternatively, index 
	\begin{equation}\label{eq:convratenomaxcons}
			\hat{R} := \left(1+\dfrac{d_{M}-d_{m}}{2}\right)^{\rho} \geq 1
		\end{equation}
		can be used for the same purpose, as $\hat{R} \propto \hat{r}$. From \eqref{eq:convratenomaxcons}, the quantity
		$R = \phi \hat{R}$ in \eqref{eq:Rconvergence} is finally suggested.
\end{remark}
\begin{remark}
	The detrending operation applied on the initial conditions helps reducing the value of $\overline{r}$ in \eqref{eq:upperboundr}, as $\overline{r} \propto \left\|x(0)\right\|_{\infty}$. Therefore, the adoption of Asm. \ref{ass} also allows to (likely) improve the convergence rate of Alg. \ref{alg:1}.
\end{remark}

\section*{Acknowledgements}
The authors are pleased to acknowledge the role of the Direction of Consorzio Bonifica Veneto Orientale in supporting the project and that of the technical staff in 
providing the topological and morphological data of the Cavallino OCN.

\bibliographystyle{IEEEtran}
\bibliography{fullbiblio} 
\vspace{-19pt}
\newpage

\begin{IEEEbiography}[{\includegraphics[width=1in,height=1.25in,clip,keepaspectratio]{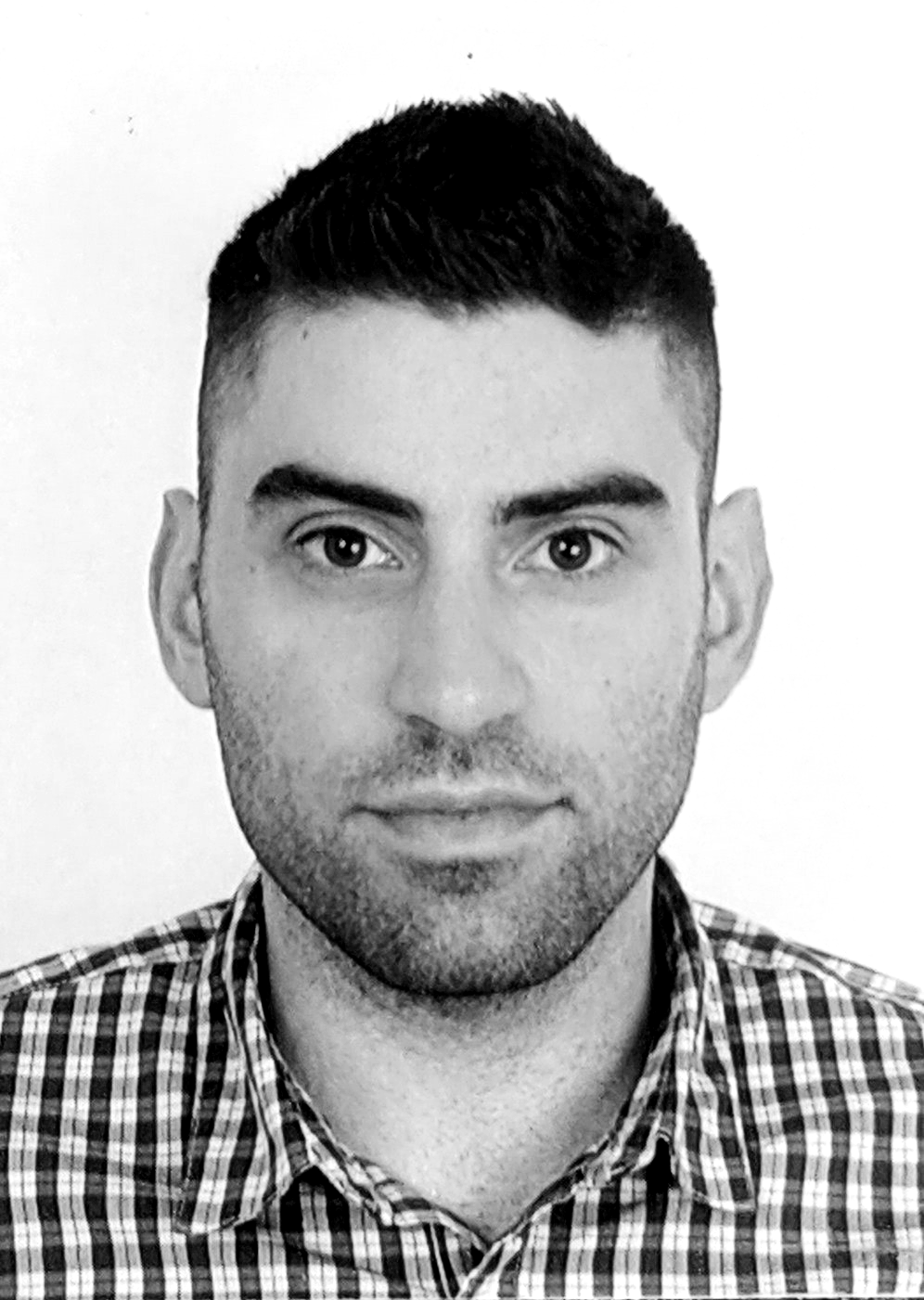}}]{Marco Fabris} received the Laurea (M.Sc.)
	degree (with honors) in Automation Engineering
	and his PhD both from the University
	of Padua, in 2016 and 2020, respectively.
	In 2018, he spent six months at the
	University of Colorado Boulder, USA, as a
	visiting scholar, focusing on distance-based
	formation tracking. In 2020-2021, he was post-doctoral fellow at the Technion-Israel Institute of Technology, Haifa. From January 2022 to July 2023 he was post-doctoral fellow at the University of Padua. He is now currently research fellow with the University of Padua and his
	main research interests involve graph-based consensus theory, resilient networks,
	optimal decentralized control and estimation for networked
	systems, data-driven predictive control and mobility as a service.
	\vspace{-12pt}
\end{IEEEbiography}

\begin{IEEEbiography}
	[{\includegraphics[width=1.05in,height=1.5in,clip,keepaspectratio]{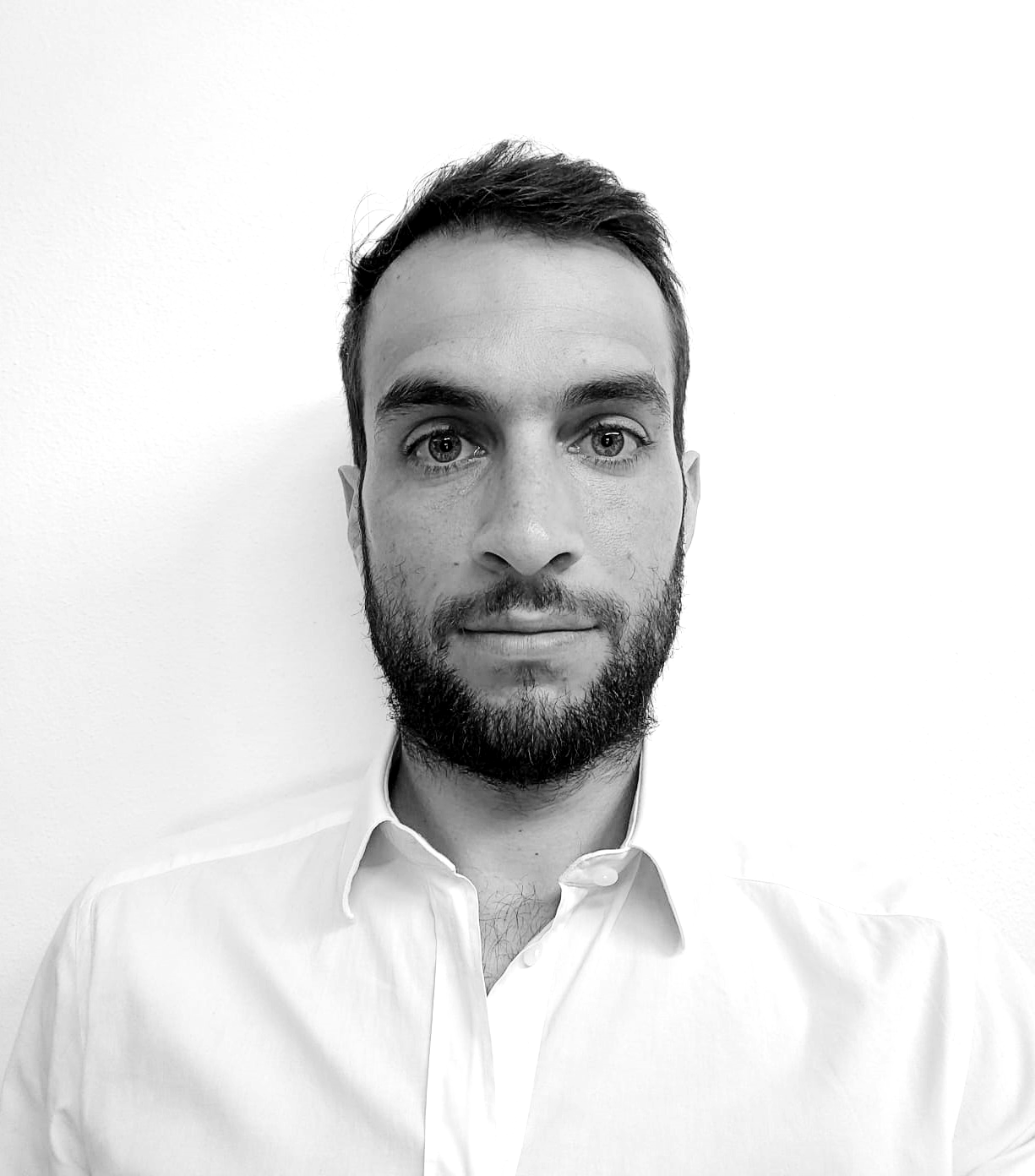}}]{Marco Davide Bellinazzi} received two Laurea (M.Sc.) degrees both from the University of Padua, one in Sport Sciences and the other on Automation Engineering, respectively in 2014 and 2022. The efforts spent during his last master thesis have contributed to achieve novel results with regard to the decentralized control over water networks by employing time-varying consensus techniques.
	\vspace{-12pt}
\end{IEEEbiography}

\vspace{-9pt}
\begin{IEEEbiography}[{\includegraphics[width=1.05in,height=1.5in,clip,keepaspectratio]{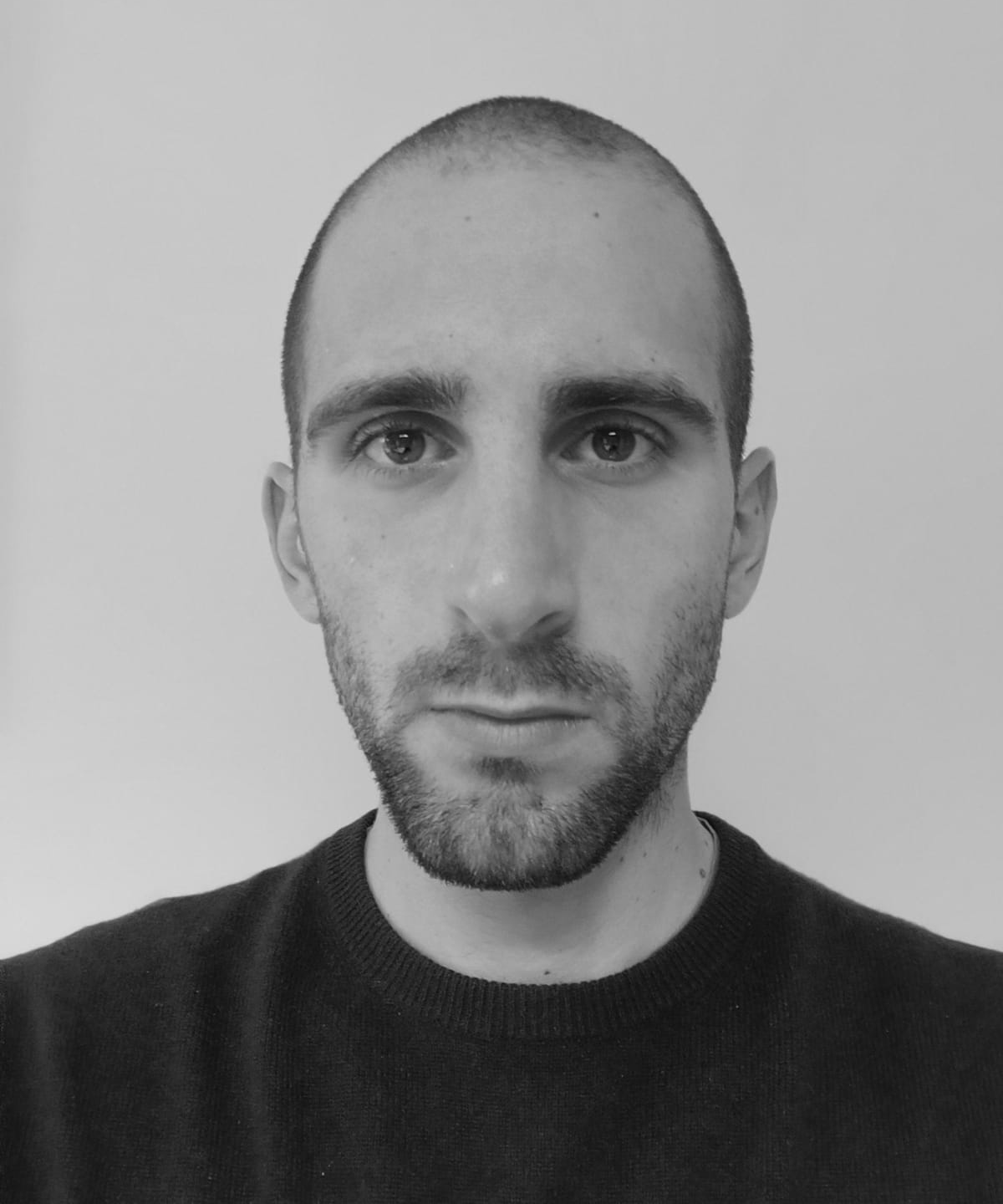}}]{Andrea Furlanetto} received the Laurea (M.Sc.) degree in Civil Engineering and his Second Level Master in GIScience and UAV, both from the University of Padua, in 2017 and 2020, respectively. Since 2019, he has worked with Consorzio di Bonifica Veneto Orientale, as a trainee first and as an employee then, focusing on remote sensing analysis and working on projects related to land reclamation and irrigation, especially for the topographic surveys aspects.
	\vspace{-12pt}
\end{IEEEbiography}


\vspace{-9pt}
\begin{IEEEbiography}[{\includegraphics[width=1in,height=1.25in,clip,keepaspectratio]{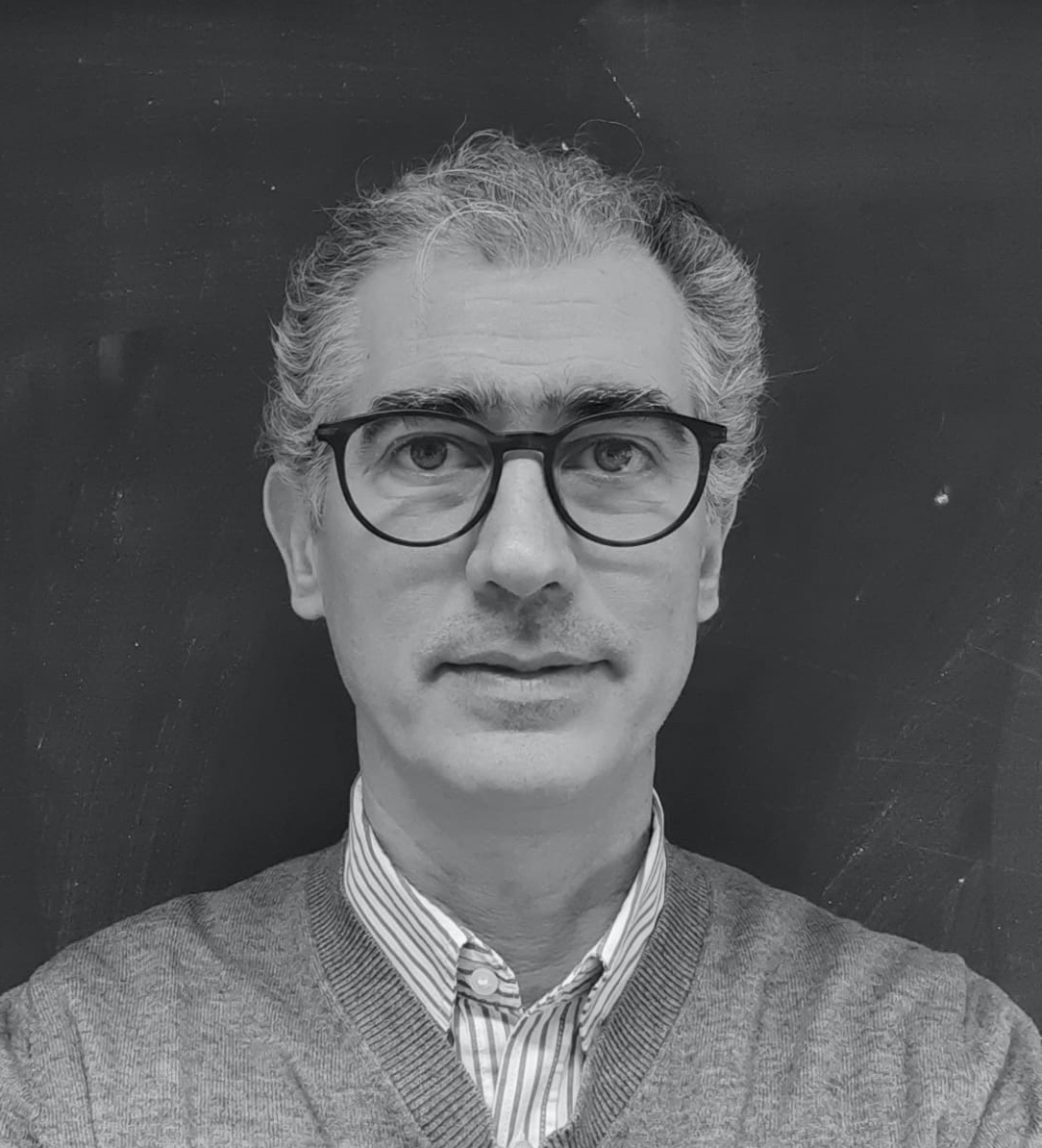}}]{Angelo Cenedese} received the M.Sc. (1999) and the Ph.D. (2004) degrees from the University of Padova, Italy, where he is currently an Associate Professor with the Department of Information Engineering. He is founder and leader of the research group SPARCS 
	and he has been and is involved in several projects on control of complex systems, funded by European and Italian government institutions and industries.
	His research interests include system modeling, control theory and its applications, mobile robotics, multi agent systems. On these subjects, he has published more than 180 papers and holds three patents.
\end{IEEEbiography}

\EOD

\end{document}